\documentclass[a4paper,11pt]{article}

\usepackage{a4wide}
\usepackage[utf8]{inputenc}
\usepackage[font=footnotesize,labelfont=bf]{caption}
\usepackage[font=scriptsize,labelfont=bf]{subcaption}

\usepackage{amsthm}
\usepackage{amsmath}
\usepackage{amssymb}
\usepackage{commath}
\usepackage{mathtools}
\usepackage{bm}
\usepackage{tikz}
\usepackage{algorithmic}
\usetikzlibrary{shapes,patterns,arrows,graphs,decorations.pathreplacing,calc}

\tikzstyle{treejob} = [shape=circle, minimum width=.4cm, inner sep=0pt, draw=black]
\tikzstyle{inner} = [shape=circle, minimum width=.3cm, inner sep=0pt, draw=none, fill = white]
\tikzstyle{block} = [pattern= north east lines]

\usepackage{xcolor}

\usepackage{graphicx,color}
\usepackage{paralist}
\usepackage{xspace}
\usepackage{hyperref}
\usepackage{cleveref}

\usepackage{booktabs}
\usepackage{multirow}
\usepackage{array}
\usepackage{listings}
\usepackage[inline]{enumitem}
\usepackage{cite} 

\usepackage{thmtools}
\usepackage{thm-restate}

\newtheorem{theorem}{Theorem}
\newtheorem{lemma}{Lemma}

\newtheorem{corollary}{Corollary}

\newtheorem{fact}{Fact}

\newcommand{\assign}{\ensuremath{\leftarrow}}

\newcommand{\I}{\mathcal{I}}
\newcommand{\G}{\mathcal{G}}
\newcommand{\jobs}{\mathcal{J}}
\newcommand{\C}{\mathcal{C}}
\newcommand{\OO}{\mathcal{O}}
\newcommand{\N}{\mathbb{N}}
\newcommand{\R}{\mathbb{R}}

\newcommand{\opt}{\textsc{Opt}\xspace}
\newcommand{\alg}{\ensuremath{\mathcal{A}}\xspace}

\newcommand{\eps}{\varepsilon}
\newcommand{\fii}{\varphi}
\newcommand{\epsdeltafrac}{\frac{\eps}{\eps-\delta}}

\newcommand{\barj}{{J_i}}
\newcommand{\bari}{\mathcal{I}}
\newcommand{\barx}{X_i} 
\newcommand{\txi}{\ensuremath{\tau_{x,t}}}
\newcommand{\txinext}{\ensuremath{\tau_{x,t+1}}}
\newcommand{\yxi}{\ensuremath{Y_{x,t}}}
\newcommand{\yxinext}{\ensuremath{Y_{x,t+1}}}
\newcommand{\thresh}{{u_t}} 
\newcommand{\prevthresh}{{u_{t-1}}} 

\newcommand{\newjob}{\ensuremath{{j^{\star}}}}

\newcommand{\region}{region algorithm\xspace}

\newcounter{conditions}
\makeatletter
\newcommand\myitem[1][]{\item[#1]\refstepcounter{conditions}\def\@currentlabel{#1}}
\makeatother

\newcommand{\hyprefR}{\ref{en:R}\xspace}
\newcommand{\hyprefC}{\ref{en:C}\xspace}
\newcommand{\hyprefP}{\ref{en:P}\xspace}

\newcommand{\hyprefLR}{\ref{en:LR}\xspace}
\newcommand{\hyprefLC}{\ref{en:LC}\xspace}
\newcommand{\hyprefLP}{\ref{en:LP}\xspace}

\definecolor{ubred}{RGB}{182,18,49}

\definecolor{utgreen}{RGB}{52,178,51}
\definecolor{utmagenta}{RGB}{207,0,114}
\definecolor{utyellow}{RGB}{254,209,0}
\definecolor{utblueORIGINAL}{RGB}{99,177,229}
\definecolor{utblue}{RGB}{0,156,235}  
\definecolor{utforest}{RGB}{0,106,77}
\definecolor{utmoos}{RGB}{0,106,77}
\definecolor{utpurple}{RGB}{79,45,127}
\definecolor{utnavy}{RGB}{10, 81, 163} 
\definecolor{utgrey}{RGB}{97,82,88}
\definecolor{utorange}{RGB}{236,122,8}

\def\schedjob#1#2#3#4#5{
	\fill[fill=#5] (#1,#3) rectangle (#2,#4);
}

\def\schedinterval#1#2#3#4#5{	
	\fill[#5!5] (#1,#3) rectangle (#2,#4);
	\draw[very thick, #5] (#1,#3+.2) to (#1,#3-.2);
	\draw[very thick, #5] (#2,#3+.2) to (#2,#3-.2);
	\draw[very thick, #5] (#1, #3) to (#2, #3);
 }

\def\blockinterval#1#2#3#4#5{
	\fill[draw=none, opacity = .8, white] (#1,#3) rectangle (#2,#4);
 	\draw[pattern=north east lines, pattern color=#5, draw=none] (#1,#3) rectangle (#2,#4);
}

\title{Online Throughput Maximization on Unrelated Machines: Commitment is No Burden}

\author{
\and Franziska Eberle\thanks{Department of Mathematics, London School of Economics and Political Sciences, UK. Email: \texttt{franziska.eberle@posteo.de}.}
\and Nicole Megow\thanks{Faculty for Mathematics and Computer Science, University of Bremen, Germany. Email: \texttt{nicole.megow@uni-bremen.de}. Partially supported by the German Science Foundation (DFG) under contract  ME 3825/1.}
\and Kevin Schewior\thanks{Departement of Mathematics and Computer Science, University of Cologne, Germany. Email: \texttt{kschewior@gmail.com}. Partially supported by the DAAD within the PRIME program using funds of BMBF and the EU Marie Curie Actions.}
}

\date{\today}

\begin{document}
\maketitle
 \begin{abstract}
We consider a fundamental online scheduling problem in which jobs with processing times and deadlines arrive online over time at their release dates. The task is to determine a feasible preemptive schedule on a single or multiple possibly unrelated
machines that maximizes the number of jobs that complete before their deadline. Due to strong impossibility results for competitive analysis on a single machine, 
we require that jobs contain some {\em slack} $\varepsilon>0$, which means that the feasible time window for scheduling a job is at least $1+\varepsilon$ times its processing time on each eligible machine. Our contribution is two-fold: (i) We give the first non-trivial online algorithms for throughput maximization on unrelated machines, and (ii), this is the main focus of our paper,
we answer the question on how to handle commitment requirements which enforce that a scheduler has to guarantee at a certain point in time the completion of admitted jobs. This is very relevant, e.g., in providing cloud-computing services, and disallows last-minute rejections of critical tasks.
We present an algorithm for unrelated machines that is $\Theta\big(\frac{1}\varepsilon\big )$-competitive when the scheduler must commit upon starting a job.
Somewhat surprisingly, this is the same optimal performance bound (up to constants) as for scheduling without commitment on a single machine. If commitment decisions must be made before a job's slack becomes less than a $\delta$-fraction of its size, we prove a competitive ratio of 
$\mathcal{O}\big(\frac{1}{\varepsilon - \delta}\big)$
for $0 < \delta < \varepsilon$. This result nicely interpolates between commitment upon starting a job and commitment upon arrival. For the latter commitment model, it is known that no (randomized) online algorithm admits any bounded competitive ratio. While we mainly focus on scheduling without migration, our results also hold when comparing against a migratory optimal solution in case of identical machines.
\end{abstract}

\section{Introduction} 

We consider the following online scheduling problem: there are given~$m$
unrelated parallel machines. Jobs from an unknown job set arrive online over time at their {\em release dates}~$r_j$. Each job~$j$ has a {\em deadline}~$d_j$ and a {\em processing time} $p_{ij}\in\mathbb{R}_+\cup\{\infty\}$, which is the execution time of $j$ when processing on machine $i$; both job parameters become known to an algorithm at job arrival. We denote a machine~$i$ with~$p_{ij}<\infty$ as {\em eligible} for job~$j$. If all machines are identical,~$p_{ij}=p_j$ holds for every job~$j$, and we omit the index~$i$. are~$m$ identical parallel machines to process these jobs or a subset of them. 
When scheduling these jobs or a subset of them, we allow {\em preemption}, i.e., the processing of a job can be interrupted at any time and may resume later without any additional cost. We mainly study scheduling without {\em migration} which means that a job must run completely on one machine. In case that we allow migration, a preempted job can resume processing on any machine, but no job can run simultaneously on two or more machines.

In a feasible schedule, two jobs are never processing at the same time on the same machine. A job is said to {\em complete} if it receives $p_{ij}$ units of processing time on machine~$i$
within the interval~$[r_j,d_j)$ if~$j$ is processed by machine~$i$.
The number of completed jobs in a feasible schedule is called {\em throughput}. The task is to find a feasible schedule with maximum throughput. We refer to this problem as {\em throughput maximization}.

As jobs arrive online and scheduling decisions are irrevocable, we cannot hope to find an optimal schedule even when scheduling on a single machine~\cite{DertouzosM89}. 
To assess the performance of online algorithms, we resort to standard {\em competitive analysis}. This means, we compare the throughput of an online algorithm with the throughput achievable by an optimal offline algorithm that knows the job set in advance. 

On a single machine, it is well-known that ``tight'' jobs with~$d_j - r_j \approx p_j$ prohibit competitive online decision making as jobs must start immediately and do not leave a chance for observing online arrivals~\cite{DBLP:journals/rts/BaruahKMMRRSW92}. Thus, it is commonly required that jobs contain some {\em slack}~$\eps>0$, i.e., every job~$j$ satisfies~$d_j - r_j \geq (1 + \eps) p_j$. In the more general setting with unrelated machines, we assume that each job $j$ satisfies~$d_j - r_j \geq (1 + \eps) p_{ij}$ for each machine $i$ that is eligible for~$j$, i.e., each machine $i$ with $p_{ij}<\infty$. The competitive ratio of our online algorithm will be a function of~$\eps$; the greater the slack, the better should the performance of our algorithm be. This slackness parameter has been considered in a multitude of previous work, e.g., in~\cite{LucierMNY13,AzarKLMNY15,SchwiegelshohnS16,GarayNYZ02,Goldwasser1999,BaruahH97,ChenEMSS2020}.
Other results for scheduling with deadlines use speed scaling, which can be viewed as another way to add slack to the schedule, see, e.g., \cite{BansalCP07,PruhsS10,AgrawalLLM18,ImM16,KalyanasundaramP00}.

In this paper, we focus on the question how to handle {\em commitment} requirements in online throughput maximization. Modeling commitment addresses the issue that a high-throughput schedule may abort jobs close to their deadlines in favor of many shorter and more urgent tasks~\cite{FergusonBKBF12}, which may not be acceptable for the job owner. Consider a company that starts outsourcing mission-critical processes to external clouds and that needs a guarantee that jobs complete before a certain time point when they cannot be moved to another computing cluster anymore. In other situations, a commitment to complete jobs  might be required even earlier just before starting the job, e.g., for a faultless copy of a database~\cite{ChenEMSS2020}. 

Different commitment models have been formalized~\cite{LucierMNY13,AzarKLMNY15,ChenEMSS2020}. The requirement to commit at a job's release date has been ruled out for online throughput maximization by strong impossibility results (even for randomized algorithms)~\cite{ChenEMSS2020}. 
We distinguish two commitment models.
\begin{enumerate}
	\item[(i)] {\em Commitment upon job admission}: an algorithm may discard a job any time before its start, we say its admission. This reflects a situation such as the faultless copy of a database.  
	\item[(ii)] {\em $\delta$-commitment}: given~$0 < \delta < \eps$, an algorithm must commit to complete a job while the job's remaining slack is at least a $\delta$-fraction of its original processing time. This models an early enough commitment (parameterized by $\delta$) for mission-critical jobs. For identical parallel machines, the latest time for committing to job~$j$ is then $d_j - (1 + \delta) p_j$. When given unrelated machines, such a commitment model might be arguably less relevant. We consider it only for non-migratory schedules and include also the choice of a processor in the commitment; we define the latest time point for committing to job~$j$ as $d_j - (1 + \delta) p_{ij}$ when processing~$j$ on machine $i$.
\end{enumerate}

Recently, a first unified approach has been presented for these models for a single machine~\cite{ChenEMSS2020}. In this and other works~\cite{LucierMNY13,AzarKLMNY15}, there remained gaps in the performance bounds and it was left open whether scheduling with commitment is even ``harder'' than without commitment. Moreover, it remained unsettled whether the problem is tractable on multiple identical or even heterogeneous machines.

In this work, we give tight results for online throughput maximization on unrelated parallel machines and answer the ``hardness'' question to the negative. 
We give an algorithm that achieves the provably best competitive ratio (up to constant factors) for the aforementioned commitment models. Somewhat surprisingly, we show that the same competitive ratio of $\OO\big(\frac{1}{\eps}\big)$ can be achieved for both, scheduling {\em without} commitment and {\em with} commitment upon admission. For unrelated machines, this is the first nontrivial result for online throughput maximization
with and without commitment. For identical parallel machines, this is the first online algorithm with bounded competitive ratio for arbitrary slack parameter~$\eps$. Interestingly, for this machine environment, our algorithm does not require job migration in order to be competitive against a migratory algorithm.

\subsection{Related work} Preemptive online scheduling and admission control have been studied rigorously. There are several results regarding the impact of deadlines on online scheduling; see, e.g.,~\cite{DBLP:conf/rtss/BaruahHS94,GarayNYZ02,Goldwasser1999} and references therein. In the following we give an overview of the literature focused on (online) throughput maximization. 

\paragraph{Offline scheduling.}
In case that the jobs and their characteristics are known to the scheduler in advance, the notion of commitment is irrelevant as an offline algorithm only starts jobs that will be completed on time; there is no benefit in starting jobs without completing them. The offline problem is well understood: For throughput maximization on a single machine, there is a polynomial-time algorithm by Lawler~\cite{Lawler1990}. The more general model, in which each job $j$ has a {\em weight} $w_j$ and the task is to maximize the total weight of jobs completed on time ({\em weighted throughput}), is NP-hard, and we do not expect polynomial time algorithms. The algorithm by Lawler solves this problem optimally in time~$\OO(n^5 w_{\max})$, where~$w_{\max} = \max_j w_j$, and can be used to design a fully polynomial-time approximation scheme (FPTAS)~\cite{PruhsW07}.

When given multiple identical machines, (unweighted) throughput maximization becomes NP-hard even for identical release dates~\cite{Lawler82}. Kalyanasundaram and Pruhs~\cite{KalyanasundaramP01} show a $6$-approximate reduction to the single-machine problem which implies a $(6+\delta)$-approximation algorithm for weighted throughput maximization on identical parallel machines, for any $\delta > 0$, using the FPTAS for the  single-machine problem~\cite{PruhsW07}. Preemptive throughput maximization on unrelated machines is much less understood from an approximation point of view. The problem is known to be strongly NP-hard~\cite{DuL91}, even without release dates~\cite{Sitters2005}. We are not aware of any approximation results for preemptive throughput maximization on unrelated machines. The situation is different for non-preemptive scheduling. In this case, throughput maximization is MAX-SNP hard~\cite{Bar-NoyGNS01} and several approximation algorithms for this general problem as well as for identical parallel machines and other special cases are known; see, e.g.,~\cite{Bar-NoyGNS01,berman,ImLM20}.

\paragraph{Online scheduling without commitment.}
For single-machine throughput maximization, Baruah, Haritsa, and Sharma~\cite{DBLP:conf/rtss/BaruahHS94} show that, in general, no deterministic online algorithm achieves a bounded competitive ratio. Thus, their result justifies our assumption on~$\eps$-slackness of each job. Moreover, they consider special cases such as unit-size jobs or agreeable deadlines where they provide constant-competitive algorithms even without further assumptions on the slack of the jobs. Here, deadlines are agreeable if~$r_j \leq r_{j'}$ for two jobs~$j$ and~$j'$ implies~$d_j \leq d_{j'}$. In our prior work~\cite{ChenEMSS2020}, we improve upon previous bounds~\cite{KalyanasundaramP00,LucierMNY13} for general instances by giving an~$\OO\big(\frac{1}{\eps}\big)$-competitive algorithm, and we show an asymptotically matching lower bound for deterministic algorithms.

For maximizing  weighted throughput, Lucier et al.~\cite{LucierMNY13} give an~$\OO\big(\frac{1}{\eps^2}\big)$-competitive online algorithm for scheduling on identical parallel machines. In a special case of this problem, called {\em machine utilization} the goal is to maximize the total processing time of completed jobs, i.e., for $p_j=w_j$ for any job $j$. This problem is much more tractable. On a single machine, Baruah et al.~\cite{DBLP:conf/focs/BaruahKMRRS91,DBLP:journals/rts/BaruahKMMRRSW92} provide a best-possible online algorithm achieving a competitive ratio of~$4$, even without any slackness assumptions. Baruah and Haritsa~\cite{BaruahH97} are the first to investigate the problem under the assumption of~$\eps$-slack and give a~$\frac{1+\eps}{\eps}$-competitive algorithm which is asymptotically best possible. For parallel identical machines (though without migration), DasGupta and Palis~\cite{DBLP:conf/approx/DasGuptaP00} give a simple greedy algorithm that achieves the same performance guarantee of~$\frac{1+\eps}{\eps}$ and give an asymptotically matching lower bound. Schwiegelshohn and Schwiegelshohn~\cite{SchwiegelshohnS16} show that migration helps an online algorithm and improves the competitive ratio to $\OO\big(\sqrt[m]{1/\eps}\big)$ for~$m$ machines.

In a line of research without slackness assumption, Baruah et al.~\cite{DBLP:conf/focs/BaruahKMRRS91} show a lower bound of~${(1+\sqrt{k})^2}$ for deterministic single-machine algorithms, where~$k = \frac{\max_j w_j/p_j}{\min_j w_j/p_j}$ is the \emph{importance ratio} of a given instance. Koren and Shasha give a matching upper bound~\cite{DBLP:journals/siamcomp/KorenS95} and generalize it to~$\Theta(\ln k)$ for parallel machines if~$k>1$~\cite{DBLP:journals/tcs/KorenS94}. Interestingly, Kalyanasundaram and Pruhs~\cite{KalyanasundaramP03} give a \emph{randomized}~$\OO(1)$-competitive algorithm for throughput maximization on a single machine without slackness assumption, which is in stark contrast to the impossibility result for deterministic single-machine algorithms.
	Very recently, Moseley et al.~\cite{MoseleyPSZ21} give a remarkable $\OO(1)$-competitive \emph{deterministic} algorithm for $m\geq 2$ parallel identical machines.

\paragraph{Online scheduling with commitment upon job arrival.} In our prior work~\cite{ChenEMSS2020}, we rule out bounded competitive ratios for any (even randomized) online algorithm for throughput maximization with commitment upon job arrival, even on a single machine. Previously, such impossibility results where only shown exploiting
	weights~\cite{LucierMNY13}.

Again, the special case~$w_j = p_j$, or machine utilization, is much more tractable than weighted or unweighted throughput maximization. A simple greedy algorithm already achieves the best possible competitive ratio~$\frac{1+\eps}{\eps}$ on a single machine, even for commitment upon arrival, as shown by DasGupta and Palis~\cite{DBLP:conf/approx/DasGuptaP00} and the matching lower bound by Garay et al.~\cite{GarayNYZ02}. For scheduling with commitment upon arrival on~$m$ parallel identical machines, there is an $\OO(\sqrt[m]{1/\eps})$-competitive algorithm and an almost matching lower bound by Schwiegelshohn and Schwiegelshohn~\cite{SchwiegelshohnS16}
Suprisingly, this model also allows for bounded competitive ratios when preemption is not allowed. In this setting, Goldwasser and Kerbikov~\cite{GoldwasserK03} give a best possible~$\big(2 + \frac1\eps\big)$-competitive algorithm on a single machine.
Very recently, Jamalabadi, Schwiegelshohn, and Schwiegelshohn~\cite{JamalabadiSS20} extend this model to parallel machines; their algorithm is near optimal with a performance guarantee approaching~$\ln \frac1\eps$ as~$m$ tends to infinity. 

\paragraph{Online scheduling with commitment upon admission and ${\delta}$-commitment.} In our previous work~\cite{ChenEMSS2020}, we design an online single-machine algorithm, called the \emph{region algorithm}, that simultaneously (with the respective choice of parameters) achieves the first non-trivial upper bounds for both commitment models. For commitment
upon job admission, our bound is~$\OO\big(\frac 1{\eps^2}\big)$, and in the $\delta$-commitment model it is $\OO\big(\frac{\eps}{(\eps -\delta)\delta^2}\big)$, for $0<\delta<\eps$. For scheduling on identical parallel machines and commitment upon admission, 
 Lucier et al.~\cite{LucierMNY13} give a heuristic that empirically performs very well but for which they cannot show a rigorous worst-case bound. In fact, Azar et al.~\cite{AzarKLMNY15} show that no bounded competitive ratio is possible for weighted throughput maximization for small~$\eps$. For~$\delta = \frac\eps2$ in the {$\delta$-commitment model}, they design (in the context of truthful mechanisms) an online algorithm for weighted throughput maximization that is~$\Theta\big( \frac{1}{\sqrt[3]{1+\eps} -1} + \frac{1}{(\sqrt[3]{1+\eps} -1)^2} \big)$-competitive if the slack~$\eps$ is sufficiently large, i.e., if~$\eps > 3$.  
 For weighted throughput, this condition on the slack is necessary as is shown by a strong general lower bound, even on a single machine~\cite{ChenEMSS2020}. For the unweighted setting, we give the first rigorous upper bound for arbitrary $\eps$ in this paper for both models, commitment upon admission and $\delta$-commitment, in the identical and even in the unrelated machine environment.

Machine utilization is again better understood. As commitment upon arrival is more restrictive than commitment upon admission and $\delta$-commitment, the previously mentioned results immediately carry over and provide bounded competitive ratios. 

\subsection{Our results and techniques}

Our main result is an algorithm that computes a non-migratory schedule that is best possible (up to constant factors) for online throughput maximization with and without commitment on identical parallel machines and, more generally, on unrelated machines. This is the first non-trivial online result for unrelated machines and it closes gaps for identical parallel machines. Our algorithm is universally applicable (by setting parameters properly) to both commitment models as well es scheduling without commitment. 

\begin{theorem}
	\label{theo:com:UB}
	Consider throughput maximization on 
	unrelated machines
	without migration. There is an $\OO\big(\frac{1}{\eps-\delta'}\big)$-competitive non-migratory online algorithm for scheduling with commitment, where $\delta'= \frac \eps 2$ in the model with commitment upon admission and $\delta' = \max\{\delta,\frac{\eps}{2}\}$ in the $\delta$-commitment model. 
\end{theorem}

For scheduling with commitment upon admission, 
this is (up to constant factors) an optimal online algorithm with competitive ratio~$\Theta\big(\frac1\eps\big)$, matching the lower bound of $\Omega\big(\frac{1}{\eps}\big)$ for $m=1$~\cite{ChenEMSS2020}. For scheduling with~$\delta$-commitment, our result interpolates between the models with commitment upon starting a job and commitment upon arrival. If~$\delta \leq \frac \eps 2$, the competitive ratio is~$\Theta\big(\frac1\eps\big)$, which is again best possible~\cite{ChenEMSS2020}. For~$\delta \rightarrow \eps$, the commitment requirements essentially implies commitment upon job arrival which has unbounded competitive ratio~\cite{ChenEMSS2020}. 

In our analysis, we compare a non-migratory schedule, obtained by our algorithm, with an optimal non-migratory schedule. However, in the case of identical machines the throughput of an optimal migratory schedule can only be larger by a constant factor than the throughput of an optimal non-migratory schedule. In fact, Kalyanasundaram and Pruhs~\cite{KalyanasundaramP01} showed that this factor is at most $\frac{6m-5}{m}$. Thus, the competitive ratio for our non-migratory algorithm, when  applied to identical machines, holds (up to this constant factor) also in a migratory setting.

\begin{corollary}
	Consider throughput maximization with or without migration on parallel identical machines. There is an $\OO\big(\frac{1}{\eps-\delta'}\big)$-competitive non-migratory online algorithm for scheduling with commitment, where $\delta'= \frac \eps 2$ in the model with commitment upon admission and $\delta' = \max\{\delta,\frac{\eps}{2}\}$ in the $\delta$-commitment model. 
\end{corollary}

The challenge in online scheduling with commitment is that, once we committed to
complete a job, the remaining slack of this job has to be spent very carefully. The key component is a job admission scheme which is implemented by different parameters. The high-level~objectives~are: 
\begin{enumerate}
	\item never start a job for the first time if its remaining slack is too small (parameter~$\delta$),
	\item during the processing of a job, admit only significantly shorter jobs (parameter $\gamma$), and 
	\item for each admitted shorter job, block some time period~(parameter~$\beta$) during which no other jobs of similar size are accepted.
\end{enumerate}

While the first two goals are quite natural and have been used before in the single and identical machine setting~\cite{LucierMNY13,ChenEMSS2020}, the third goal is crucial for our new tight result.
The intuition is the following: Think of a single eligible machine in a non-migratory schedule. Suppose we committed to complete a job with processing time~$1$ and have only a slack of $\OO(\eps)$ left before the deadline of this job. Suppose that $c$ substantially smaller jobs of size~$\frac1c$ arrive where $c$ is the competitive ratio we aim for. On the one hand, if we do not accept any of them, we cannot hope to achieve $c$-competitiveness. On the other hand, accepting too many of them fills up the slack and, thus, leaves no room for even smaller jobs. The idea is to keep the flexibility for future small jobs by only accepting an~$\eps$-fraction of jobs of similar size (within a factor two).

We distinguish two time periods that guide the acceptance decisions. During the {\em scheduling interval} of a job~$j$, we have a more restrictive acceptance scheme that ensures the completion of~$j$ whereas in the {\em blocking period} we guarantee the completion of previously accepted jobs. We call our algorithm {\em blocking} algorithm. This acceptance scheme is much more refined than the one of the known region algorithm in~\cite{ChenEMSS2020} that uses one long region with a uniform acceptance threshold and is then too conservative in accepting jobs.

Given that we consider the non-migratory version of the problem, a generalization from a single to multiple machines seems natural. It is interesting, however, that such a generalization works, essentially on a per-machine basis, even for unrelated machines and comes at no loss in the competitive ratio.

Clearly, scheduling with commitment is more restrictive than without commitment. Therefore, our algorithm is also $O\big(\frac{1}{\eps}\big)$-competitive for 
maximizing the throughput on unrelated machines without any commitment requirements. Again, this is optimal (up to constant factors)  as it matches the lower bound on the competitive ratio for deterministic online algorithms on a single machine~\cite{ChenEMSS2020}.

\begin{corollary}
	There is a~$\Theta\big(\frac1\eps\big)$-competitive algorithm for online throughput maximization on 
	unrelated machines without commitment requirements and
	without migration.
\end{corollary}

However, for scheduling without commitment, we are able to generalize the simpler region algorithm presented for the single-machine problem in~\cite{ChenEMSS2020} to scheduling on unrelated machines. 

\begin{theorem}\label{theo:otm:UB-no}
	A generalization of the region algorithm is~$\Theta\big(\frac1\eps\big)$-competitive for online throughput maximization on unrelated machines without commitment requirements and
	without migration. 
\end{theorem}

Besides presenting a simpler algorithm for throughput maximization without commitment, we show this result to present an additional application of our technical findings for the analysis of the blocking algorithm. We give details later. On a high level, we show a key lemma on the size of non-admitted jobs for a big class of online algorithms which results in an upper bound on the throughput of an optimal (offline) non-migratory algorithm. This key lemma can be used in the analysis of both algorithms, blocking and region. In fact, also the analysis of the original region algorithm for a single machine~\cite{ChenEMSS2020} becomes substantially easier.

In case of identical machines, again, we can apply the result by Kalyanasundaram and Pruhs~\cite{KalyanasundaramP01} that states that the throughput of an optimal migratory schedule is  larger by at most a constant factor than the throughput of an optimal non-migratory schedule. Thus, the result in \Cref{theo:otm:UB-no} holds also in a migratory setting when scheduling on identical machines.

\begin{corollary}
	A generalization of the region algorithm is~$\Theta\big(\frac1\eps\big)$-competitive for online throughput maximization on multiple identical machines without commitment requirements, with and without migration. 
\end{corollary}

\subsection*{Outline of the paper}  In \Cref{sec:blocking}, we describe and outline the analysis of our new non-migratory algorithm. It consists of two parts, which are detailed in \Cref{sec:com:CompleteAll,sec:otm:AdmitMany}: firstly, we argue that all jobs admitted by our algorithm can complete by their deadline and, secondly, we prove that we admit ``sufficiently many'' jobs. In \Cref{sec:RegionAlg}, we generalize the known region algorithm, developed for a single machine in our prior work~\cite{ChenEMSS2020}, to a non-migratory algorithm without commitment on unrelated machines. We show how to apply a new key technique developed for the analysis in \Cref{sec:otm:AdmitMany} to analyze it and prove the same competitive ratio (up to constant factors) as for a single machine.

\section{The blocking algorithm}
\label{sec:blocking}

In this section, we describe the \emph{blocking algorithm} for scheduling with commitment. We assume that the slackness constant $\eps >0$ and, in the $\delta$-commitment model, $\delta \in (0,\eps)$ are given. If $\delta$ is not part of the input or if~$\delta \leq \frac\eps2$, then we set $\delta = \frac{\eps}{2}$.

The algorithm never migrates jobs between machines, i.e., a job is only processed by the machine that initially started to process it. In this case, we say the job has been  \emph{admitted} to this machine. Moreover, our algorithm commits to completing a job upon admission (even in the~$\delta$-commitment model). Hence, its remaining slack has to be spent very carefully on admitting other jobs  to still be competitive. As our algorithm does not migrate jobs, it transfers the admission decision to the shortest admitted and not yet completed job on each machine. A job only admits significantly shorter jobs and prevents the admission of too many jobs of similar size. To this end, the algorithm maintains two types of intervals for each admitted job, a \emph{scheduling interval} and a \emph{blocking period}. A job can only be processed in its scheduling interval. Thus, it has to complete in this interval while admitting other jobs. Job~$j$ only admits jobs that are smaller by a factor of at least~$\gamma =\tfrac{\delta}{16} < 1$. For an admitted job~$k$, job~$j$ creates a blocking period of length at most~$\beta p_{ik}$, where~$\beta = \tfrac{16}{\delta}$, which blocks the admission of similar-length jobs (cf. Figure~\ref{fig:atree}). The scheduling intervals and blocking periods of jobs admitted by~$j$ will always be pairwise disjoint and completely contained in the scheduling interval of~$j$.

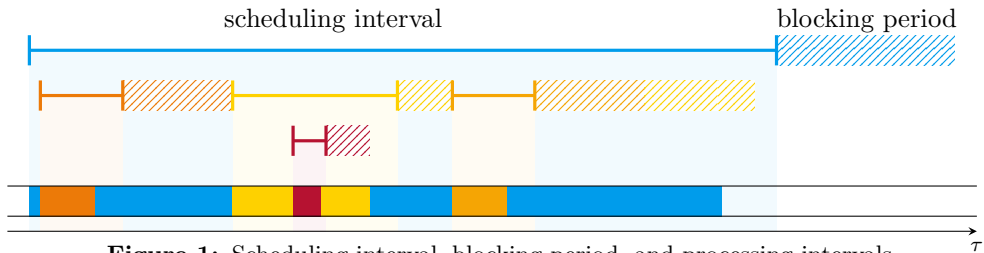
\begin{figure}[tbh]
	\centering%
	\begin{tikzpicture}[x=.85cm]
	
	\def \xscale {.85}

	\blockinterval{14.66}{14*\xscale}{2.2}{2.6}{utblue}
	\schedinterval{.4*\xscale}{14*\xscale}{2.4}{0}{utblue}
	
	\foreach \x/\col/\lb in {2.1/utorange/2, 7.1/utyellow/1, 9.6/utyellow!50!utorange/2, 11.6/utyellow/2}{
		\blockinterval{\x*\xscale}{\x*\xscale + \lb*\xscale}{1.6}{2}{\col}
	}

	\foreach \x/\col/\ls in {.6/utorange/1.5,  4.1/utyellow/3, 8.1/utyellow!50!utorange/1.5}{
		\schedinterval{\x*\xscale}{\x*\xscale+\ls*\xscale}{1.8}{0}{\col}
	}
	
	\blockinterval{5.8*\xscale}{6.6*\xscale}{1.0}{1.4}{ubred}
	\schedinterval{5.2*\xscale}{5.8*\xscale}{1.2}{0}{ubred}

	\foreach \x/\y/\col in {.4/13/utblue}{
		\schedjob{\x*\xscale}{\y*\xscale}{.2}{.6}{\col}	
	}
	
	\foreach \x/\y/\col in {.6/1.6/utorange, 4.1/6.6/utyellow, 8.1/9.1/utyellow!50!utorange}{
		\schedjob{\x*\xscale}{\y*\xscale}{.2}{.6}{\col}
	}

	\foreach \x/\y/\col in {5.2/5.7/ubred}{
		\schedjob{\x*\xscale}{\y*\xscale}{.2}{.6}{\col}	
	}
	
	\node[black, inner sep = 0pt, font = \small] at (5.04, 2.8) {scheduling interval};
	\node[black, inner sep = 0pt, font = \small] at (13.3, 2.8) {blocking period};

	\draw[-stealth] (0,0) to (15,0) node [below,font=\scriptsize] {$\tau$};
	\draw[] (0,.2) to (15,.2);
	\draw[] (0,.6) to (15,.6);

	\end{tikzpicture}%
	\vspace{-.6cm}%
	\caption{Scheduling interval, blocking period, and processing intervals}\label{fig:atree}	
\end{figure}

\paragraph{Scheduling jobs.} Independent of the admission scheme, the blocking algorithm follows the \textsc{Shortest Processing Time} (SPT) order for the set of uncompleted jobs assigned to a machine.
SPT ensures that a job~$j$ has highest priority in the blocking periods of any job~$k$ admitted by~$j$.

\paragraph{Admitting jobs.} The algorithm keeps track of \emph{available} jobs at any time point~$\tau$. A job~$j$ with $r_j \leq \tau$ is called available for machine~$i$ if it has not yet been admitted to a machine by the algorithm and its deadline is not too close, i.e., $d_j - \tau \geq (1+\delta) p_{ij}$. 

Whenever a job~$j$ is available for machine~$i$ at a time~$\tau$ 
such that time~$\tau$ is not contained in the scheduling interval of any other job admitted to~$i$, the shortest such job~$j$ is immediately admitted to machine~$i$ at time~$a_j := \tau$, creating the scheduling interval $S(j) = [a_j, e_j)$, where~$e_j = a_j + (1+\delta)p_{ij}$ and an empty blocking period~$B(j) = \emptyset$. In general, however, the blocking period of a job $j$ is a finite union of time intervals associated with~$j$, and its size is the sum of lengths of the intervals, denoted by~$|B(j)|$. Both, blocking period and scheduling interval, depend on machine $i$ but we omit~$i$ from the notation as it is clear from the context; both periods are created after job~$j$ has been assigned to machine~$i$.

Four types of events trigger a decision of the
algorithm at time $\tau$: the release of a job, the end of a blocking period, the end of a scheduling interval, and the admission of a job. In any of these four cases, the algorithm calls the {\em admission routine}.
This subroutine iterates over all machines~$i$ and checks if~$j$, the shortest job on~$i$ whose scheduling interval contains~$\tau$, can admit the currently shortest job~$\newjob$ available for machine~$i$.

To this end, any admitted job~$j$ checks whether~$p_{i\newjob} < \gamma p_{ij}$. Only such jobs qualify for admission by~$j$. Upon admission by~$j$, job~$\newjob$ obtains two disjoint consecutive intervals, the \emph{scheduling interval} $S(\newjob) = [a_{\newjob}, e_{\newjob})$ and the \emph{blocking period} $B(\newjob)$ of size at most~$\beta p_{i\newjob}$. At the admission of job~$\newjob$, the blocking period~$B(\newjob)$ is planned to start at~$e_{\newjob}$, the end of~$\newjob$'s scheduling interval. 
During $B(\newjob)$, job~$j$ only admits jobs~$k$ with~$p_{ik} < \frac12 p_{i\newjob}$. 

Hence, when job~$j$ decides if it admits the currently shortest available job~$\newjob$ at time~$\tau$, it makes sure that~$\newjob$ is sufficiently small and that no job~$k$ of similar (or even smaller) processing time is blocking~$\tau$, i.e., it verifies that~$\tau \notin B(k)$ for all jobs~$k$ with~$p_{ik} \leq 2 p_{i\newjob}$ admitted to the same machine. 
In this case, we say that~$\newjob$ is a \emph{child} of~$j$ and that~$j$ is the \emph{parent} of~$\newjob$, denoted by $\pi(\newjob) = j$. 
If job~$\newjob$ is admitted at time~$\tau$ by job~$j$, the algorithm sets~$a_{\newjob} =\tau$ and~$e_{\newjob} = a_{\newjob} + (1+\delta) p_{i\newjob}$ and assigns the scheduling interval~$S(\newjob) = [a_{\newjob}, e_{\newjob})$ to~${\newjob}$. 

If~$e_{\newjob} \leq e_j$, the routine sets~$f_{\newjob} = \min\{e_j, e_{\newjob} + \beta p_{i\newjob}\}$ which
determines~$B(\newjob) = [e_{\newjob}, f_{\newjob})$. As the scheduling and blocking periods of children~$k$ of~$j$ are supposed to be disjoint, 
we have to \textit{update the blocking periods}. First consider the job~$k$ with~$p_{ik} > 2 p_{i\newjob}$ admitted to the same machine whose blocking period contains~$\tau$ (if it exists), and let~$[e_{k}', f_{k}')$ be the maximal interval of~$B(k)$ containing~$\tau$. We set $f_{k}''= \min \{e_j, f_{k}' + (1+\delta + \beta)p_{i\newjob}\}$ and replace the interval~$[e_{k}', f_{k}')$  by $[e_{k} ', \tau) \cup [\tau+(1+\delta + \beta) p_{i\newjob}, f_{k}'')$. For all other jobs~$k$ with~$B(k) \cap [\tau,\infty) \neq \emptyset$ admitted to the same machine, we replace the remaining part of their blocking period~$[e_k', f_k')$ by~$[e_k'+(1+\delta+\beta)p_{i\newjob}, f_k'')$ where~$f_k'' := \min\{e_j, f_k'+(1+\delta+\beta)p_{i\newjob} \} $. In this update, we follow the convention that $[e,f) = \emptyset$ if $f \leq e$. Observe that the length of the blocking period might decrease due to such updates. 

Note that~$e_{\newjob} > e_j$ is also possible as~$j$ does not take the end of its own scheduling interval~$e_j$ into account when admitting jobs. Thus, the scheduling interval of~$\newjob$ would end outside the scheduling interval of~$j$ and inside the blocking period of~$j$. During $B(j)$, the parent~$\pi(j)$ of~$j$, did not allocate the interval~$[e_j,e_{\newjob})$ for completing jobs admitted by~$j$ but for ensuring its own completion. Hence, the completion of both~$\newjob$ and~$\pi(j)$ is not necessarily guaranteed anymore. To prevent this, we \textit{modify all scheduling intervals}~$S(k)$ (including~$S(j)$) that contain time~$\tau$ of jobs admitted to the same machine as~$\newjob$ and their blocking periods~$B(k)$. For each job~$k$ admitted to the same machine with~$\tau\in S(k)$ (including~$j$) and~$e_{\newjob} > e_k$, we set~$e_k = e_{\newjob}$. We also update their blocking periods (in fact, single intervals) to reflect their new starting points. If the parent~$\pi(k)$ of~$k$ does not exist,~$B(k)$ remains empty; otherwise we set~$B(k) := [e_k, f_k)$ where~$f_k = \min\{e_{\pi(k)}, e_k + \beta p_{ik}\}$. Note that, after this update, the blocking periods of any but the largest such job will be empty. Moreover, the just admitted job~$\newjob$ does not get a blocking period in this special case. 

During the analysis of the algorithm, we show that any admitted job~$j$ still completes before~$a_j + (1+\delta)p_{ij}$ and that~$e_j \leq a_j + (1+2\delta)p_{ij}$ holds in retrospect for all admitted jobs~$j$. Thus, any job~$j$ that admits another job~$\newjob$ tentatively assigns this job a scheduling interval of length $(1+\delta)p_{i\newjob}$ but, for ensuring its own completion, it is prepared to lose~$(1+2\delta)p_{i\newjob}$ time units of its scheduling interval~$S(j)$. 
We summarize the blocking algorithm in Figure~\ref{alg:BlockAlg}.

\begin{figure}
\caption{Blocking algorithm}
\label{alg:BlockAlg}
\begin{algorithmic}
\STATE{\textbf{Scheduling Routine:} At all times $\tau$ and on all machines $i$, run the job with shortest processing time that has been admitted to $i$ and has not yet completed. \smallskip}
\STATE{\textbf{Event:} Release of a new job at time $\tau$}
\STATE{\hspace{1em}Call Admission Routine.\smallskip} 	
\STATE{\textbf{Event:} End of a blocking period or scheduling interval at time $\tau$}
\STATE{\hspace{1em}Call Admission Routine.\smallskip} 
\STATE{\textbf{Admission Routine:}}
\STATE{$i \assign 1$ }
\STATE{$\newjob \assign$ a shortest job available at $\tau$ for machine $i$, i.e., $\newjob \in \arg\min\{p_{ij} \,|\, j \in \jobs, r_j \leq \tau \text{ and } d_j - \tau \geq (1 + \delta) p_{ij}\}$}
\WHILE{$i \leq m$}
	\STATE{$K$ $\assign$ the set of jobs on machine $i$ whose scheduling intervals contain $\tau$ }
	\IF{$K = \emptyset$}
		\STATE{admit job $\newjob$ to machine $i$, $a_{\newjob} \leftarrow \tau$, and $e_{\newjob} \leftarrow a_{\newjob} + (1+\delta) p_{i\newjob}$ }
		\STATE{$S({\newjob})\leftarrow[a_{\newjob},e_{\newjob})$ and $B(\newjob) \leftarrow \emptyset$}
		\STATE{call Admission Routine}
	\ELSE
		\STATE{$j \assign \arg\min\{p_{ik} \,|\, k \in K\}$}
		\IF{$\newjob < \gamma p_{ij}$ \AND $\tau \notin B(j')$ for all $j'$ admitted to $i$ with $ p_{ij'} \leq 2 p_{i\newjob}$}
			\STATE{admit job $\newjob$ to machine $i$, $a_{\newjob} \leftarrow \tau$, and $e_{\newjob} \leftarrow a_{\newjob} + (1+\delta) p_{i\newjob}$ }
			\IF{$e_{\newjob} \leq e_j$}
				\STATE{$f_{\newjob} \leftarrow \min\{e_j, e_{\newjob} + \beta p_{i\newjob}\}$}
				\STATE{$S({\newjob})\leftarrow[a_{\newjob},e_{\newjob})$ and $B(\newjob) \leftarrow [e_{\newjob}, f_{\newjob})$}
			\ELSE
				\STATE{$S({\newjob})\leftarrow[a_{\newjob},e_{\newjob})$ and $B(\newjob) \leftarrow \emptyset$}
				\STATE{modify $S(k)$ and $B(k)$ for $k \in K$}
				\STATE{update $B(j')$ for $j'$ admitted to machine $i$ with $B(j') \cap [\tau,\infty) \neq \emptyset$}
				\STATE{call Admission Routine}
			\ENDIF
		\ELSE
			\STATE{$i \assign i+1$}
			\STATE{$\newjob \assign$ a shortest job available at $\tau$ for machine $i$, i.e., $\newjob \in \arg\min\{p_{ij} \,|\, j \in \jobs, r_j \leq \tau \text{ and } d_j - \tau \geq (1 + \delta) p_{ij}\}$}
		\ENDIF
	\ENDIF
\ENDWHILE
\end{algorithmic}
\end{figure}

\subsection*{Roadmap for the analysis}

During the analysis, 
it is sufficient to concentrate on instances with small slack, as also noted in \cite{ChenEMSS2020}. For $\eps > 1$ we run the blocking algorithm with $\eps = 1$, which only tightens the commitment requirement, and obtain constant competitive ratios. Thus, we assume $0 < \eps \leq 1$. For $0 < \delta < \eps$, in the~$\delta$-commitment model an online scheduler needs to commit to the completion of a job~$j$ no later than~$d_j - (1+\delta)p_{ij}$. Hence,  committing to the completion of a job~$j$ at an earlier point in time clearly satisfies committing at a remaining slack of~$\delta p_{ij}$. Therefore, we may assume $\delta \in [\frac\eps2, \eps)$. 

The blocking algorithm does not migrate any job. In the analysis, we compare the throughput of our 
algorithm to the solution of an optimal non-migratory schedule. To do so, we rely on a key design principle of the blocking algorithm, which is that, whenever the job set admitted to a machine is fixed, the scheduling of the jobs follows the simple SPT order. This enables us to split the analysis into two parts. 

In the first part, we argue that the scheduling routine can handle the admitted jobs sufficiently well. That is, every admitted jobs is completed on time; see \Cref{sec:com:CompleteAll}. Here, we use that the blocking algorithm is non-migratory and consider each machine individually.

For the second part, we observe that the potential admission of a new job~$\newjob$ to machine~$i$ is solely based on its availability and on its size relative to~$j_i$, the job currently processed by machine~$i$. More precisely, given the availability of~$\newjob$, if~$p_{i\newjob} < \gamma p_{ij_i}$, the time does not belong to the blocking period of a job~$k_i$ admitted to machine~$i$ with~$p_{i\newjob} \geq \frac 12 p_{ik_i}$ and~$i$ is the first machine (according to machine indices) with this property, then~$\newjob$ is admitted to machine~$i$. This implies that~$\min \big\{ \gamma p_{ij_i},  \frac12 p_{ik_i}\big\} $ acts as a \emph{threshold}, and only available jobs with processing time less than this threshold qualify for admission by the blocking algorithm on machine~$i$. Hence, any available job that the blocking algorithm does not admit has to exceed the threshold. 

Based on this observation, we develop a general charging scheme for \emph{any} non-migratory online algorithm satisfying the property that, at any time~$\tau$, the algorithm maintains a time-dependent threshold and the shortest available job smaller than this threshold is admitted by the algorithm. We formalize this description and analyze the competitive ratio of such algorithms in \Cref{sec:otm:AdmitMany} before applying this general result to our particular algorithm. 

\section{Completing all admitted jobs on time}\label{sec:com:CompleteAll}

We show that the blocking algorithm finishes every admitted job on time in \Cref{theo:com:CompleteAll}. 

\begin{theorem}\label{theo:com:CompleteAll}
	Let $0<\delta < \eps$ be fixed. If $0 < \gamma < 1$ and $\beta \geq 1$ satisfy 
	\begin{equation}\label{eq:CompleteAll}
	\frac{\beta/2}{\beta/2 + (1+2\delta)}\left( 1 + \delta - 2(1+2\delta)\gamma\right) \geq 1,
	\end{equation}  	
	then the blocking algorithm completes any job~$j$ admitted at~$a_j\leq d_j - (1+\delta)p_{ij}$ on time.
\end{theorem}
Recall that we chose $\gamma = \frac \delta {16}$ and $\beta = \frac{16}\delta$, which guarantees that Equation~\eqref{eq:CompleteAll} is satisfied. 

As the blocking algorithm does not migrate jobs, it suffices to consider each machine individually in this section. 
The proof relies on the following observations: (i) The sizes of jobs admitted by job~$j$ that interrupt each others' blocking periods are geometrically decreasing, (ii) the scheduling intervals of jobs are completely contained in the scheduling intervals of their parents, and (iii) scheduling in SPT order guarantees that job~$j$ has highest priority in the blocking periods of its children. We start by proving the following technical lemma about the length of the final scheduling interval of an admitted job~$j$, denoted by~$|S(j)|$. In the proof, we use that~$\pi(k)=j$ for two jobs~$j$ and~$k$ implies that~$p_{ik} < \gamma p_{ij}$. 

\begin{lemma} \label{lem:com:LengthSchedulingInterval}
	Let $0<\delta < \eps$ be fixed. If $\gamma > 0$ satisfies $(1+2\delta)\gamma \leq \delta,$
	then~$|S(j)| \leq (1+2\delta)p_{ij}$. Moreover, $S(j)$ contains the scheduling intervals and blocking periods of all descendants of~$j$. 
\end{lemma}
\renewcommand{\newjob}{k}
\begin{proof}
	Consider a machine~$i$ and let~$j$ be a job admitted to machine~$i$.
	By definition of the blocking algorithm, the end point~$e_j$ of the scheduling interval of job~$j$ is only modified when~$j$ or one of~$j$'s descendants admits another job. Let us consider such a case: If job~$j$ admits a job~$\newjob$ whose scheduling interval does not fit into the scheduling interval of~$j$, we set~$e_j = e_{\newjob}= a_{\newjob} + (1+\delta)p_{i\newjob}$ to accommodate the scheduling interval~$S(\newjob)$ within~$S(j)$. The same modification is applied to any ancestor~${j'}$ of~$j$ with $e_{j'} < e_{\newjob}$. This implies that, after such a modification of the scheduling interval, neither~$j$ nor any affected ancestor~${j'}$ of~$j$ are the smallest jobs in their scheduling intervals anymore. In particular, no job whose scheduling interval was modified in such a case at time~$\tau$ is able to admit jobs after~$\tau$. Hence, any job~$j$ can only admit other jobs within the interval $[a_j, a_j + (1+\delta)p_{ij})$. That is, $a_{\newjob} \leq a_j + (1+\delta)p_{ij}$ for every job~$\newjob$ with~$\pi(\newjob) = j$.
	
	Thus, by induction, it is sufficient to show that~$a_{\newjob} + (1+2\delta)p_{i\newjob} \leq a_j + (1+2\delta)p_{ij}$ for admitted jobs~$\newjob$ and~$j$ with~$\pi(\newjob) = j$. Note that~$\pi(\newjob) = j$ implies $p_{i\newjob} < \gamma p_{ij}$. Hence, 
	\[
	a_{\newjob} + (1+2\delta) p_{i\newjob} \leq (a_j + (1+\delta)p_{ij}) + (1+2\delta)\gamma p_{ij} \leq a_j + (1+2\delta)p_{ij}, 
	\] 
	where the last inequality follows from the assumption $(1+2\delta) \gamma \leq \delta$. Due to the construction of~$B(k)$ upon admission of some job~$k$ by job~$j$, we also have~$B(k) \subseteq S(j)$. 
\end{proof}
\renewcommand{\newjob}{\ensuremath{j^\star}}

\begin{proof}[Proof of \Cref{theo:com:CompleteAll}]
	Let~$j$ be a job admitted by the blocking algorithm to machine~$i$ with $a_j\leq d_j - (1+\delta)p_{ij}$. Showing that job~$j$ completes before time~$d_j' := a_j + (1+\delta)p_{ij}$ proves the theorem. Due to scheduling in SPT order, each job~$j$ has highest priority in its own scheduling interval if the time point does not belong to the scheduling interval of a descendant of~$j$. Thus, it suffices to show that at most~$\delta p_{ij}$ units of time in $[a_j, d'_j)$ belong to scheduling intervals~$S(k)$ of descendants of~$j$. By \Cref{lem:com:LengthSchedulingInterval}, the scheduling interval of any descendant~$k'$ of a child~$k$ of~$j$ is contained in~$S(k)$. Hence, it is sufficient to only consider~$K$, the set of children of~$j$.
	
	In order to bound the contribution of each child~$k\in K$, we impose a \emph{class structure} on the jobs in~$K$ depending on their size relative to job~$j$. 
	More precisely, we define~$(\C_c(j))_{c\in\N_0}$, where~$\C_c(j)$ contains all jobs~$k \in K$ that satisfy $\frac{\gamma}{2^{c+1}} p_{ij} \leq p_{ik} < \frac{\gamma}{2^c} p_{ij}$. As~$k \in K$ implies~$p_{ik} < \gamma p_{ij}$, each child of~$j$ belongs to exactly one class and~$(\C_c(j))_{c\in\N_0}$ in fact partitions~$K$.  
	
	Consider two jobs~$k,k'\in K$ where, upon admission,~$k$ interrupts the blocking period of~$k'$. By definition, we have~$p_{ik} < \frac12 p_{ik'}$. Hence, the chosen class structure ensures that~$k$ belongs to a strictly \emph{higher} class than~$k'$, i.e., there are~$c,c' \in \N$ with~$c>c'$ such that~$k \in \C_c(j)$ and~$k'\in \C_{c'}(j)$. In particular, the admission of a job~$k \in \C_c(j)$ implies either that~$k$ is the first job of class~$\C_c(j)$ that~$j$ admits or that the blocking period of the previous job in class~$\C_c(j)$ has completed. Based on this distinction, we are able to bound the loss of scheduling time for~$j$ in~$S(j)$ due to~$S(k)$ of a child~$k$.	
	Specifically, we partition~$K$ into two sets. The first set~$K_1$ contains all children of~$j$ that where admitted as the first jobs in their class~$\C_c(j)$. The set~$K_2$ contains the remaining jobs.
	
	We start with $K_2$. Consider a job~$k \in \C_c(j)$ admitted by~$j$. By \Cref{lem:com:LengthSchedulingInterval}, we know that~$|S(k)| = (1+\mu \delta)p_{ik}$, where $1\leq \mu \leq 2$.  Let~$k' \in \C_c(j)$ be the previous job admitted by~$j$ in class~$\C_c(j)$. Then, $B(k') \subseteq [e_{k'}, a_{k})$. Since scheduling and blocking periods of children of~$j$ are disjoint,~$j$ has highest scheduling priority in~$B(k')$. Hence, during $B(k') \cup S(k)$ job~$j$ is processed for at least~$|B(k')|$ units of time. In other words,~$j$ is processed for at least a~$\frac{|B(k')|}{|B(k') \cup S(k)|}$-fraction of $B(k') \cup S(k)$. We rewrite this ratio as 
	\begin{equation*}
	\frac{|B(k')|}{|B(k') \cup S(k)|} = \frac{\beta p_{ik'}}{\beta p_{ik'} + (1+\mu\delta) p_{ik}} = \frac{\nu \beta}{\nu \beta + (1+\mu\delta)},
	\end{equation*}
	where $\nu := \frac{p_{ik'}}{p_{ik}} \in (\frac{1}{2}, 2]$. By differentiating with respect to $\nu$ and $\mu$, we observe that the last term is increasing in~$\nu$ and decreasing in~$\mu$. Thus, we lower bound this expression by 
	\begin{equation*}
	\frac{|B(k')|}{|B(k') \cup S(k)|} \geq \frac{\beta/2}{\beta/2 + (1+2\delta)}.
	\end{equation*}
	Therefore, $j$ is processed for at least a $\frac{\beta/2}{\beta/2 + (1+2\delta)}$-fraction in $\bigcup_{k \in K}B(k) \cup \bigcup_{k \in K_2}S(k)$.
	
	We now consider the set~$K_1$. The total processing volume of these jobs is bounded from above by \(
	\sum_{c=0}^\infty \frac{\gamma }{2^c} p_{ij} = 2\gamma p_{ij}. 	
	\)
	By \Cref{lem:com:LengthSchedulingInterval},~$|S(k)| \leq (1+2\delta)p_{ik}$. Combining these two observations, we obtain
	\(
	\big| \bigcup_{k \in K_1 } S(k) \big| \leq 2(1+2\delta)\gamma p_{ij}.
	\)
	Combining
	the latter with the bound for~$K_2$, we conclude that~$j$ is scheduled for at least 
	\begin{equation*}
	\Big|[a_j, d_j')\setminus \bigcup_{k \in K} S(k) \Big| \geq \frac{\beta/2}{\beta/2 + (1+2\delta)} \big( (1+\delta) - 2(1+2\delta)\gamma \big) p_{ij} \geq p_{ij}
	\end{equation*}
	units of time, where the last inequality follows from Equation \eqref{eq:CompleteAll}. Therefore,~$j$ completes before~$d_j' = a_j + (1+\delta)p_{ij} \leq d_j$, which concludes the proof.
\end{proof}

\section{Competitiveness: admitting sufficiently many jobs}\label{sec:otm:AdmitMany}

This section shows that the blocking algorithm admits sufficiently many jobs to be~$\OO\big(\frac1{\eps-\delta}\big)$-com\-pet\-i\-tive. As mentioned before, this proof is based on the observation that, at time~$\tau$, the blocking algorithm admits any job available for machine~$i$ if its processing time is less than~$\gamma p_{ij_i}$, where~$j_i$ is the job processed by machine~$i$ at time~$\tau$, and this time point is not blocked by another job~$k_i$ previously admitted by~$j_i$ to machine~$i$. We start by formalizing this observation for a class of non-migratory online algorithms before proving that this enables us to bound the number of jobs any feasible schedule successfully schedules during a particular period. Then, we use it to show that the blocking algorithm is indeed~$\OO\big(\frac1{\eps-\delta }\big)$-competitive. 

\subsection{A class of online algorithms}\label{sec:otm:ClassAlgorithms} 

In this section, we investigate a class of non-migratory online algorithms. 
Recall that a job~$j$ is called available for machine~$i$ at time~$\tau$ if it is  released before or at time~$\tau$,~$d_j - \tau \geq  (1+\delta)p_{ij}$, and is not yet admitted. 

We consider a non-migratory online algorithm \alg with the following properties.  
\begin{enumerate}[label = (P\arabic*)]
	\item\label{enum:alg:available} \alg only admits available jobs. 
	\item\label{enum:alg:threshold} Retrospectively, for each time~$\tau$ and each machine~$i$, there is a threshold~$u_{i\tau} \in [0, \infty]$ such that any job~$j$ that was available for machine~$i$ and not admitted to machine~$i$ by \alg at time~$\tau$ satisfies~$p_{ij} \geq u_{i\tau}$. 
	The function~$u^{(i)} : \R \rightarrow [0,\infty],  \tau \mapsto u_{i\tau}$ is piece-wise constant and right-continuous for every machine~$i \in \{1,\ldots,m\}$. Further, there are only countably many points of discontinuity. (This last property is used to simplify the exposition.)
\end{enumerate}

\subsubsection*{Key lemma on the size of non-admitted jobs}

For the proof of the main result in this section, we rely on the following strong, structural lemma about the volume processed by a feasible non-migratory schedule in some time interval and the size of jobs admitted by a non-migratory online algorithm satisfying \ref{enum:alg:available} and \ref{enum:alg:threshold} in the same time interval.

Let~$\sigma$ be a feasible non-migratory schedule. Without loss of generality, we assume that~$\sigma$ completes all jobs that it started on time. Let~$X^\sigma$ be the set of jobs completed by $\sigma$ and not admitted by \alg. For~$1\leq i \leq m$, let~$X_i^\sigma$ be the set of jobs in~$X^\sigma$ processed by machine~$i$. Let~$C_x$ be the completion time of job~$x\in X^\sigma$ in $\sigma$. 

\begin{lemma}\label{lem:Volume} 
	Let $0 \leq \vartheta_1 \leq \vartheta_2$ and fix $x\in X_i^\sigma$ as well as $Y\subset X_i^\sigma\setminus \{x\}$. If 
	\begin{enumerate}
		\myitem[(R)]\label{en:LR} $r_x \geq \vartheta_1$ as well as $r_y \geq \vartheta_1$ for all $y \in Y$,
		\myitem[(C)]\label{en:LC} $C_x\geq C_y$ for all $y \in Y$, and 
		\myitem[(P)]\label{en:LP} $\sum_{y \in Y} p_{iy} \geq \frac{\eps}{\eps - \delta} (\vartheta_2 - \vartheta_1)$
	\end{enumerate}
	hold, then $p_{ix} \geq u_{i\vartheta_2}$, where~$u_{i\vartheta_2}$ is the threshold imposed by~\alg at time~$\vartheta_2$. In particular, if~$u_{i,\vartheta_2} = \infty$, then no such job~$x$ exists.
\end{lemma}

\begin{proof}	
	We show the lemma by contradiction.	More precisely, we show that, if~$p_{ix} < u_{i\vartheta_2}$, the schedule~$\sigma$ cannot complete~$x$ on time and, hence, is not feasible. 
	
	Remember that~$x\in X_i^\sigma$ implies that \alg did not admit job~$x$ at any point~$\vartheta$. At time~$\vartheta_2$, there are two possible reasons why~$x$ was not admitted: $p_{ix} \geq u_{i\vartheta_2}$ or $ d_x - \vartheta_2 < (1+\delta)p_{ix}$. In case of the former, the statement of the lemma holds. Toward a contradiction, suppose~$p_{ix} < u_{i\vartheta_2}$ and, thus, $d_x - \vartheta_2 < (1+\delta)p_{ix}$ has to hold. As job~$x$ arrives with a slack of at least~$\eps p_{ix}$ at its release date~$r_x$ and~$r_x \geq \vartheta_1$ by assumption, we have
	\begin{equation}\label{eq:zeta-length}
	\vartheta_2 - \vartheta_1 \geq \vartheta_2 - d_x +d_x - r_x > -(1+\delta) p_{ix} + (1+\eps)p_{ix}= (\eps - \delta)p_{ix}. 	
	\end{equation}
	
	Since all jobs in~$Y$ complete earlier than~$x$ by Assumption~\hyprefLC and are only released after~$\vartheta_1$ by~\hyprefLR, the volume processed by~$\sigma$ in $[\vartheta_1, C_x)$ on machine~$i$ is at least $\epsdeltafrac (\vartheta_2 - \vartheta_1) + p_{ix}$ by~\hyprefLP. Moreover,~$\sigma$ can process at most a volume of $(\vartheta_2 - \vartheta_1)$ on machine~$i$ in~$[\vartheta_1, \vartheta_2)$. These two bounds imply that~$\sigma$ has to process job parts with a processing volume of at least  \[
	\epsdeltafrac (\vartheta_2 - \vartheta_1) + p_{ix} - (\vartheta_2-\vartheta_1)  > 	\frac{\delta}{\eps-\delta} (\eps - \delta) p_{ix} + p_{ix} = (1+\delta) p_{ix}
	\]
	in $[\vartheta_2,C_x)$, where the inequality follows using Inequality~\eqref{eq:zeta-length}. Thus, $C_x \geq \vartheta_2 + (1+\delta) p_{ix} > d_x$, which contradicts the feasibility of~$\sigma$. 
	
	Observe that, by~\ref{enum:alg:available} and~\ref{enum:alg:threshold}, the online algorithm \alg admits a job available for machine~$i$ if it satisfies~$p_{ij} < u_{i\tau}$. In particular, if $u_{i\tau} = \infty$ for some time point~$\tau$, then~\alg admits any job available for machine~$i$. Hence, for~$0 \leq \vartheta_1 \leq \vartheta_2$ with $u_{i\vartheta_2} = \infty$, there does not exist a job~$x \in X_i^\sigma$ and a set~$Y \subset X_i^\sigma\setminus \{x\}$ satisfying \hyprefLR, \hyprefLC, and \hyprefLP for machine~$i$.	
\end{proof}

\subsubsection*{Bounding the number of non-admitted jobs}

In this section, we use the Properties~\ref{enum:alg:available} and~\ref{enum:alg:threshold} to bound the throughput of a non-migratory optimal (offline) algorithm. To this end, we fix an instance as well as an optimal schedule with job set \opt. Let \alg be a non-migratory online algorithm satisfying~\ref{enum:alg:available} and~\ref{enum:alg:threshold}. 

Let~$X$ be the set of jobs in \opt that the algorithm \alg did not admit. We assume without loss of generality that all jobs in \opt complete on time. Since \opt as well as~\alg are non-migratory, we compare the throughput machine-wise. To this end, we fix one machine~$i$.
Let~$X_i \subset X$ be the set of jobs scheduled on machine~$i$ by \opt. 

Assumption~\ref{enum:alg:threshold} guarantees that the threshold~$u_{i,\tau}$ is piece-wise constant and right-continuous, i.e.,~$u^{(i)}$ is constant on intervals of the form~$[\tau_t, \tau_{t+1})$.  Let~$\bari$ represent the set of maximal intervals~$I_t = [\tau_t, \tau_{t+1})$ where~$u^{(i)}$ is constant. That is,~$u_{i,\tau}= \thresh$ holds for all~$\tau \in I_t$ and~$u_{i,\tau_{t+1}} \neq \thresh$, where~$\thresh:= u_{i,\tau_t}$, 
The main result of this section is the following theorem.

\begin{theorem}\label{theo:otm:Charging}
	Let~$\barx$ be the set of jobs that are scheduled on machine~$i$ in the optimal schedule. Let~$\bari = \{I_1,\ldots,I_{T}\}$ be the set of maximal intervals on machine~$i$ of \alg such that the machine-dependent threshold is constant for each interval and has the value~$\thresh$ in interval~$I_t= [\tau_{t}, \tau_{t+1})$. Then, 
	\[|\barx| \leq \sum_{t=1}^{T} \epsdeltafrac \frac{\tau_{t+1} - \tau_t}{\thresh}  + T,\]
	where we set~$\frac{\tau_{t+1} - \tau_t}{\thresh} = 0$ if~$\thresh = \infty$ and $\frac{\tau_{t+1} - \tau_t}{\thresh} = \infty$ if~$\{\tau_t, \tau_{t+1}\} \cap \{-\infty, \infty\} \neq \emptyset$ and~$\thresh < \infty$.
\end{theorem}

We observe that~$T = \infty$ trivially proves the statement as~$\barx$ contains at most finitely many jobs. The same is true if~$\frac{\tau_{t+1} - \tau_t}{\thresh} = \infty$ for some~$t \in [T]$. Hence, for the remainder of this section we assume without loss of generality that~$\bari$ only contains finitely many intervals and that~$\frac{\tau_{t+1} - \tau_t}{\thresh} < \infty$ holds for every~$t \in [T]$. 

To prove this theorem, we develop a charging scheme that assigns jobs~$x \in \barx$ to intervals in~$\bari$. 
The idea behind our charging scheme is that \opt does not contain arbitrarily many jobs that are available in~$I_t$ since~$\thresh$ provides a natural lower bound on their processing times. In particular, the processing time of any job that is \emph{released} during interval~$I_t$ and not admitted by the algorithm exceeds the lower bound~$\thresh$. Thus, the charging scheme relies on the release date~$r_x$ and the size~$p_{ix}$ of a job~$x\in \barx$ as well as on the precise structure of the intervals created by \alg. 

The charging scheme we develop is based on a careful modification of the following partition~$(F_t)_{t=1}^T$ of the set~$\barx$. Fix an interval~$I_t \in \bari$ and define the set~$F_t \subseteq \barx$ as the set that contains all jobs~$x\in \barx$ released during~$I_t$, i.e., $F_t = \{ x \in \barx: r_x \in I_t \}$. 
Since, upon release, each job~$x \in \barx$ is available and not admitted by~\alg, the next fact directly follows from Properties~\ref{enum:alg:available} and~\ref{enum:alg:threshold}.
\begin{fact}\label{cor:Ft}
	For all jobs~$x \in F_t$ it holds~$p_{ix} \geq \thresh$. In particular, if~$\thresh = \infty$, then~$F_t = \emptyset$.
\end{fact}

In fact, the charging scheme maintains this property and only assigns jobs in~$\barx$ to intervals~$I_t$ if $p_{ix} \geq \thresh$. In particular, no job will be assigned to an interval with~$\thresh = \infty$.

We now formalize how many jobs in~$\barx$ are assigned to a specific interval~$I_t$. 
Let \[
\fii_t := \Big\lfloor \epsdeltafrac \frac{\tau_{t+1} - \tau_t}{\thresh} \Big\rfloor + 1 
\]
if~$\thresh < \infty$, and~$\fii_t = 0$ if~$\thresh = \infty$. We refer to~$\fii_t$ as the \emph{target number} of~$I_t$. 
As discussed before, we assume~$\frac{\tau_{t+1} - \tau_t}{\thresh} < \infty$, and, thus, the target number is well-defined.  
If each of the sets~$F_t$ satisfies~$|F_t| \leq \fii_t$, then \cref{theo:otm:Charging} immediately follows. In general, $|F_t| \leq \fii_t$ does not have to be true since jobs in~\opt may be preempted and processed during several intervals~$I_t$. Therefore, for proving \cref{theo:otm:Charging}, we show that there always exists another partition~$(G_t)_{t=1}^T$ of~$\barx$ such that~$|G_t| \leq \fii_t$ holds. 

The high-level idea of this proof is the following: Consider an interval~$I_t = [\tau_t, \tau_{t+1})$. If $F_t$ does not contain too many jobs, i.e., $|F_t| \leq \fii_t$, we would like to set $G_t = F_t$. Otherwise, we find a later interval~$I_{t'}$ with $|F_{t'}| < \fii_{t'}$ such that we can assign the excess jobs in~$F_t$ to~$I_{t'}$.

\begin{proof}[Proof of \cref{theo:otm:Charging}]
	As observed before, it suffices to show the existence of a partition $\G=(G_t)_{t=1}^T$ of~$\barx$ such that $|G_t| \leq \fii_t$ in order to prove the theorem. 
	
	In order to repeatedly apply \cref{lem:Volume}, we only assign excess jobs~$x\in F_t$ to $G_{t'}$ for~$t < t'$ if their processing time is at least the threshold of~$I_{t'}$, i.e., $p_{ix} \geq {u_{t'}}$. By our choice of parameters, a set~$G_{t'}$ with~$\fii_{t'}$ many jobs of size at least~${u_{t'}}$ ``covers'' the interval $I_{t'} = [\tau_{t'}, \tau_{t'+1})$ as often as required by~\hyprefLP in \cref{lem:Volume}, i.e., 
	\begin{equation}
	\sum_{x \in G_{t'}} p_{ix} \geq \fii_{t'}\cdot {u_{t'}} = \bigg( \bigg\lfloor \epsdeltafrac \frac{\tau_{t'+1} - \tau_{t'}}{{u_{t'}}} \bigg\rfloor + 1 \bigg) \cdot {u_{t'}} \geq \epsdeltafrac(\tau_{t'+1} - \tau_{t'}).
	\end{equation}
	
	The proof consists of two parts: the first one is to inductively (on $t$) construct the partition~$\G = ( G_t )_{t=1}^T$ of~$\barx$, where $|G_t| \leq \fii_t$ holds for~$t \in [T-1]$. The second one is the proof that a job~$x\in G_t$ satisfies $p_{ix} \geq {u_{t}}$ which will imply~$|G_T| \leq \fii_T$. During the construction of~$\G$ we define temporary sets~$A_{t}\subset \barx$ for intervals~$I_t$. The set~$G_{t}$ is chosen as a subset of~$F_{t} \cup A_{t}$ of appropriate size. In order to apply \cref{lem:Volume} to each job in $A_{t}$ individually, alongside~$A_t$, we construct a set~$Y_{x,t}$ and a time~$\tau_{x,t}\leq r_x$ for each job~$x\in \barx$ that is added to~$A_t$. Let~$C_y^*$ be the completion time of some job~$y\in \barx$ in the optimal schedule~\opt. The second part of the proof is to show that~$x$,~$\tau_{x,t}$, and~$Y_{x,t}$ satisfy  
	
	\begin{enumerate}
		\myitem[(R)]\label{en:R} $r_y \geq \txi$ for all $y\in \yxi$, 
		\myitem[(C)]\label{en:C} $C_x^* \geq C_y^*$ for all $y\in \yxi$,  and 
		\myitem[(P)]\label{en:P} $\sum_{y \in \yxi} p_{iy} \geq \epsdeltafrac(\tau_t - \txi)$. 
	\end{enumerate}
	
	This implies that $x$, $Y = \yxi$, $\vartheta_1 = \txi$, and $\vartheta_2 = \tau_t$ satisfy the conditions of \Cref{lem:Volume}, and thus the processing time of~$x$ is at least the threshold at time~$\tau_t$, i.e.,~$p_{ix}\geq u_{i\tau_t} = {u_{t}}$.
	
	\paragraph{Constructing $G = ( G_t )_{t=1}^T$.}
	We inductively construct the sets~$G_t$ in the order of their indices. We start by setting~$A_t = \emptyset$ for all intervals~$I_t$ with $t \in {T}$. We define~$\yxi = \emptyset$ for each job~$x\in \barx$ and each interval~$I_t$. The preliminary value of the time~$\txi$ is the minimum of the starting point~$\tau_t$ of the interval~$I_t$ and the release date~$r_x$ of~$x$, i.e., $\txi := \min\{\tau_t, r_x\}$. We refer to the step in the construction where~$G_t$ was defined by \emph{step}~$t$. 
	
	Starting with $t=1$, let~$I_t$ be the next interval to consider during the construction with~$t < T$. Depending on the cardinality of~$F_t \cup A_t$, we distinguish two cases. If $|F_t \cup A_t| \leq \fii_t$, then we set~$G_t = F_t \cup A_t$.
	
	If $|F_t \cup A_t| > \fii_t$, then we order the jobs in~$F_t \cup A_t$ in increasing order of completion times in the optimal schedule. The first~$\fii_t$ jobs are assigned to~$G_t$ while the remaining~$|F_t \cup A_t| - \fii_t$ jobs are added to~$A_{t+1}$. In this case, we might have to redefine the times~$\tau_{x,t+1}$ and the sets~$Y_{x,t+1}$ for the jobs~$x$ that were newly added to~$A_{t+1}$. Fix such a job~$x$. If there is no job~$z$ in the just defined set~$G_t$ that has a smaller release date than~$\txi$, we set~$\txinext= \txi$ and~$\yxinext = \yxi \cup G_t$. Otherwise let~$z \in G_t$ be a job with $r_z < \txi$ that has the smallest time~$\tau_{z,t}$. We set~$\txinext = \tau_{z,t}$ and~$\yxinext = Y_{z,t} \cup G_t$. 
	
	Finally, we set~$G_T = F_T \cup A_T$. We observe that~${u_T} < \infty$ implies~$\fii_T = \infty$ because~$\tau_{T+1} = \infty$. Since this contradicts the assumption~$\fii_t < \infty$ for all~$t \in [T]$, this implies~${u_T}  = \infty$. We will show that~$p_x\geq {u_{T}}$ for all~$x \in G_T$. Hence,~$G_T = \emptyset$. Therefore~$|G_T| = \fii_T = 0$.
	
	\paragraph{Bounding the size of jobs in~$G_t$.}
	We consider the intervals again in increasing order of their indices and show by induction that any job~$x$ in~$G_t$ satisfies~$p_{ix} \geq {u_{t}}$ which implies~$G_t = \emptyset$ if ${u_t} = \infty$. 
	Clearly, if~$x \in F_t \cap G_t$, \cref{cor:Ft} guarantees~$p_{ix} \geq \thresh$. Hence, in order to show the lower bound on the processing time of~$x \in G_t$, it is sufficient to consider jobs in~$G_t \setminus F_t \subset A_t$. To this end, we show that for such jobs \hyprefR, \hyprefC, and \hyprefP are satisfied. 
	Thus, \cref{lem:Volume} guarantees that~$p_{ix} \geq u_{i\tau_t} = \thresh$ by definition. Hence,~$A_t = \emptyset$ if~${u_{t}} = \infty$ by \cref{lem:Volume}.
	
	By construction,~$A_1 = \emptyset$. Hence, \hyprefR, \hyprefC, and \hyprefP are satisfied for each job~$x \in A_1$.
	
	Suppose that the Conditions~\hyprefR,~\hyprefC, and~\hyprefP are satisfied for all~$x \in A_s$ for all~$1 \leq s < t$. Hence, for~$s < t$, the set~$G_s$ only contains jobs $x$ with~$ p_{ix} \geq {u_{s}}$. Fix~$x \in A_t$. We want to show that~$p_{ix} \geq {u_{t}}$. By the induction hypothesis and by \cref{cor:Ft},~$p_{iy} \geq \prevthresh$ holds for all~$y \in G_{t-1}$. Since~$x$ did not fit in~$G_{t-1}$ anymore,~$|G_{t-1}| = \fii_{t-1}$. 
	
	We distinguish two cases based on~$G_{t-1}$. If there is no job~$z \in G_{t-1}$ with~$r_z < \tau_{x,t-1}$, then~$ \tau_{x,t}=\tau_{x,t-1}$, and~\hyprefR and~\hyprefC are satisfied by construction and by the induction hypothesis. For~\hyprefP, consider 
	\begin{align*}
	\sum_{y \in \yxi} p_{iy} & = \sum_{y \in Y_{x,t-1}} p_{iy} + \sum_{y \in G_{t-1}} p_{iy} \\
	& \geq \epsdeltafrac (\tau_{t-1} - \tau_{x,t-1}) + \prevthresh\cdot \fii_{t-1} \\
	& \geq \epsdeltafrac (\tau_{t-1} - \tau_{x,t-1}) + \epsdeltafrac (\tau_{t} - \tau_{t-1})\\
	& = \epsdeltafrac (\tau_{t} - \tau_{x,t})\, ,		
	\end{align*}
	where the first inequality holds due to the induction hypothesis. By \cref{lem:Volume},~$p_{ix} \geq u_{\tau_t} = \thresh$.
	
	If there is a job~$z \in G_{t-1}$ with $r_z < \tau_{x,t-1} \leq \tau_{t-1}$, then $z \in A_{t-1}$. In step~$t-1$, we chose~$z$ with minimal~$\tau_{z,t-1}$. Thus,~$r_y \geq \tau_{y,t-1} \geq \tau_{z,t-1}$ for all~$y \in G_{t-1}$ and~$r_x \geq \tau_{x,t-1} > r_z \geq \tau_{z,t-1}$ which is Condition~\hyprefR for the jobs in~$G_{t-1}$. Moreover, by the induction hypothesis,~$r_y \geq \tau_{z,t-1}$ holds for all~$y \in Y_{z,t-1}$. Thus, $\tau_{x,t}$ and~$Y_{x,t}$ satisfy~\hyprefR. For~\hyprefC, consider that~$C_x^* \geq C_y^*$ for all~$y \in G_{t-1}$ by construction and, thus, ~$C_x^* \geq C_z^* \geq C_y^*$ also holds for all~$y \in Y_{z,t-1}$ due to the induction hypothesis. For~\hyprefP, observe that
	\begin{align*}
	\sum_{y \in Y_{x,t}} p_{iy} & = \sum_{y \in Y_{z,t-1}} p_{iy} + \sum_{y \in G_{t-1}} p_{iy} \\
	& \geq \epsdeltafrac (\tau_{t-1} - \tau_{z,t-1}) + \prevthresh\cdot \fii_{t-1} \\
	& \geq \epsdeltafrac (\tau_{t-1} - \tau_{z,t-1}) + \epsdeltafrac ( \tau_{t} - \tau_{t-1}) \\
	& \geq \epsdeltafrac (\tau_t - \tau_{x,t}).		
	\end{align*}	
	Here, the first inequality follows from the induction hypothesis
	and the second from the definition of~$\prevthresh$ and~$\fii_{t-1}$. Hence, \cref{lem:Volume} implies~$p_{ix} \geq u_{\tau_t} = \thresh$.

	We note that~$p_{ix} \geq \thresh$ for all~$x \in G_t$ and for all~$t \in [T]$. 
	\paragraph{Bounding~$|\barx|$.} 
	By construction, we know that $\bigcup_{t=1}^{T} G_t = \barx$. We start with considering intervals~$I_t$ with~${u_{t}} = \infty$. Then,~$I_t$ has an unbounded threshold, i.e.,~$u_{i\tau} = \infty$ for all~$\tau \in I_t$,  and~$F_t = \emptyset$ by \cref{cor:Ft}.
	In the previous part we have seen that the conditions for \cref{lem:Volume} are satisfied. Hence,~$G_t = \emptyset$ if~${u_{t}} = \infty$. For~$t$ with~${u_{t}} < \infty$, we have~$|G_t| \leq \fii_t = \big\lfloor \epsdeltafrac \frac{\tau_{t+1} - \tau_t}{{u_{t}}} \big\rfloor + 1$. As explained before, this bounds the number of jobs in~$\barx$. 
\end{proof}

\subsection{The blocking algorithm admits sufficiently many jobs}

Having the powerful tool that we developed in the previous section at hand, it remains to show that the blocking algorithm admits sufficiently many jobs to achieve the competitive ratio of~$\OO\big( \frac{1}{\eps-\delta'}\big)$ where~$\delta' = \frac\eps2$ for commitment upon admission and~$\delta' = \max \big\{\frac\eps2, \delta \big\}$ for~$\delta$-commitment. To this end, we show that the blocking algorithm belongs to the class of online algorithms considered in \Cref{sec:otm:ClassAlgorithms}. Then, \cref{theo:otm:Charging} provides a bound on the throughput of an optimal non-migratory schedule.

We begin by showing that the blocking algorithm satisfies Properties~\ref{enum:alg:available} to~\ref{enum:alg:threshold}. The first property is clearly satisfied by the definition of the blocking algorithm. For the second and the third property, we observe that a new job~$\newjob$ is only admitted to a machine~$i$ during the scheduling interval of another job~$j$ admitted to the same machine if~$p_{i\newjob} < \gamma p_{ij}$. Further, the time point of admission must not be blocked by a similar- or smaller-size job~$k$ previously admitted during the scheduling interval of~$j$. This leads to the bound~$p_{i\newjob} < \frac12 p_{ik}$ for any job~$k$ whose blocking period contains the current time point. Combining these observations leads to a machine-dependent threshold~$u_{i,\tau} \in [0,\infty]$ satisfying~\ref{enum:alg:threshold}.

More precisely, fix a machine~$i$ and a time point~$\tau$. Using~$j \rightarrow i$ to denote that~$j$ was admitted to machine~$i$, we define~$u_{i,\tau} := \min_{j:\, j \rightarrow i, \tau \in S(j)} \gamma p_{ij} $ if there is no job~$k$ admitted to machine~$i$ with~$\tau \in B(k)$, with~$\min \emptyset = \infty$. Otherwise, we set~$u_{i,\tau} := \frac12p_{ik}$. We note that the function~$u^{(i)}$ is piece-wise constant and right-continuous due to our choice of right-open intervals for defining scheduling intervals and blocking periods. Moreover, the points of discontinuity of~$u^{(i)}$ correspond to the admission of a new job, the end of a scheduling interval, and the start as well as the end of a blocking period of jobs admitted to machine~$i$. Since we only consider instances with a finite number of jobs, there are at most finitely many points of discontinuity of~$u^{(i)}$. Hence, we can indeed apply \cref{theo:otm:Charging}. 

Then, the following theorem is the main result of this section. 

\begin{theorem}\label{theo:com:AdmitEnoughJobs}
	An optimal non-migratory (offline) algorithm can complete at most a factor $\alpha+5$ more jobs on time than admitted by the blocking algorithm, where $\alpha := \epsdeltafrac\big(2\beta + \frac{1+2\delta}{\gamma} \big)$. 
\end{theorem}

\begin{proof}
	We fix an instance and an optimal solution \opt. We use~$X$ to denote the set of jobs in \opt that the blocking algorithm did not admit. Without loss of generality, we can assume that all jobs in \opt complete on time. If~$J$ is the set of jobs admitted by the blocking algorithm, then~$X \cup J$ is a superset of the jobs successfully finished in the optimal solution. Hence, showing~$|X| \leq 	(\alpha + 4)|J|$ suffices to prove \cref{theo:com:AdmitEnoughJobs}. 

For each machine~$i$, we compare the throughput of the optimal solution to the throughput on machine~$i$ of the blocking algorithm. More precisely, let~$\barx \subseteq X $ be the jobs in \opt scheduled on machine~$i$ and let~$\barj \subseteq J$ be the jobs scheduled by the blocking algorithm on machine~$i$. With \cref{theo:otm:Charging}, we show~$|\barx| \leq(\alpha + 4) |\barj|$ to bound the cardinality of~$X$ in terms of~$|J|$.
	
	To this end, we retrospectively consider the interval structure created by the algorithm on machine~$i$. Let~$\bari$ be the set of maximal intervals~$I_t = [\tau_t, \tau_{t+1})$ such that~$u_{i,\tau} = u_{i,\tau_t}$ for all~$\tau \in I_t$. We define~$\thresh = u_{i,\tau_t}$ for each interval~$I_t$. As discussed above, the time points~$\tau_t$ for~$t \in [T]$ correspond to the admission, the end of a scheduling interval, and the start as well as the end of a blocking period of jobs admitted to machine~$1$. As the admission of a job adds at most three time points, we have that~$|\bari| \leq 3 |\barj| + 1$. 
	
	As the blocking algorithm satisfies Properties~\ref{enum:alg:available} to~\ref{enum:alg:threshold}, we can apply \cref{theo:otm:Charging} to obtain 
	\[
	|\barx| \leq \sum_{t = 1}^T \epsdeltafrac\frac{\tau_{t+1} - \tau_t}{\thresh} + |\bari| \leq \sum_{t = 1}^T \epsdeltafrac\frac{\tau_{t+1} - \tau_t}{\thresh} + (3|\barj| +1 ).
	\]
	It remains to bound the first part in terms of~$|\barj|$. If~$\thresh < \infty$, let~$j_t\in \barj$ be the \emph{smallest} job~$j$ with~$\tau_t \in S(j) \cup B(j)$. Then, at most~$\epsdeltafrac \frac{\tau_{t+1} - \tau_t}{\thresh}$ (potentially fractional) jobs will be charged to job~$j_t$ because of interval~$I_t$. By definition of~$\thresh$, we have~$\thresh = \gamma p_{ij_t}$ if~$I_t \subseteq S(j_t)$, and if~$I_t \subseteq B(j_t)$, we have~$\thresh = \frac12p_{ij_t}$. The total length of intervals~$I_t$ for which~$j=j_t$ holds sums up to at most~$(1+2\delta)p_{ij}$ for~$I_t \subseteq S(j)$ and to at most~$2\beta p_{ij}$ for~$I_t \subseteq B(j)$. Hence, in total, the charging scheme assigns at most~$\epsdeltafrac ( 2 \beta + \frac{1+2\delta}{\gamma} ) = \alpha $ jobs in~$\barx$ to job~$j \in \barj$. Therefore, 
	\begin{equation*}
	|\barx| \leq \big( \alpha + 3 \big) |\barj| +1.	
	\end{equation*}	
	If~$\barj = \emptyset$, the blocking algorithm admitted all jobs scheduled on machine~$i$ by \opt, and~$|\barx| = 0 =  |\barj|$ follows. Otherwise, $
	|\barx| \leq \big( \alpha + 4 \big) |\barj|$, and we obtain 
	\[
	|\opt| \leq |X \cup J| = \sum_{i=1}^m |\barx| + |J| \leq \sum_{i=1}^m (\alpha +4) |\barj| + |J| \leq (\alpha + 5) |J|,
	\]	
	which concludes the proof. 
\end{proof}

\subsection{Finalizing the proof of \Cref{theo:com:UB}}
\begin{proof}[Proof of \cref{theo:com:UB}]
	In \Cref{theo:com:CompleteAll} we show that the blocking algorithm completes all admitted jobs~$J$ on time. This implies that the blocking algorithm is feasible for the model commitment upon admission. As no job~$j \in J$ is admitted later than~$d_j - (1+\delta)p_{ij}$, the blocking algorithm also solves scheduling with~$\delta$-commitment.
	In \Cref{theo:com:AdmitEnoughJobs}, we bound the throughput~$|\opt|$ of an optimal non-migratory solution by $\alpha+5$ times~$|J|$, the throughput of the blocking algorithm, where $\alpha=\frac{\epsilon}{\eps -\delta}(2\beta + \frac{1+2\delta}{\gamma})$. 
	Our choice of parameters $\beta = \tfrac{16}{\delta}$ and $\gamma = \tfrac{\delta}{16}$ implies that the blocking algorithm achieves a competitive ratio of $c \in \OO \big(\frac{\eps}{(\eps -\delta)\delta} \big)$. For commitment upon arrival or for~$\delta$-commitment in the case where~$\delta \leq \frac\eps2$, we run the algorithm with~$\delta'= \frac\eps2$. Hence,~$c \in \OO\big(\frac{1}{\eps - \delta'}\big) = \OO\big(\frac1\eps\big)$. If~$\delta > \frac\eps2$, then we set~$\delta'= \delta$ in our algorithm. Thus,~$\frac{\eps}{\delta'} \in \OO(1)$ and, again,~$c \in \OO\big(\frac{1}{\eps - \delta'}\big)$.	
\end{proof}

\section{Scheduling without commitment}\label{sec:RegionAlg}

This section considers online throughput maximization without commitment requirements. We show how to exploit also in this setting our key lemma on the size of non-admitted jobs for a big class of online algorithms and the resulting upper bound on the throughput of an optimal (offline) non-migratory algorithm  from~\Cref{sec:otm:ClassAlgorithms}. 

We consider the \emph{region algorithm} that was designed by~\cite{ChenEMSS2020} for scheduling on a single machine and we generalize it to parallel identical machines. We prove that it has a competitive ratio of~$\OO\big(\frac1\eps\big)$, which is best possible on a single machine and improves substantially upon the best previously known parallel-machine algorithm (for weighted throughput) with a competitive ratio of $\OO\big(\frac{1}{\eps^2}\big)$ by Lucier et al.~\cite{LucierMNY13}. For a single machine, this  matches the guarantee proven in~\cite{ChenEMSS2020}. However, our new analysis is much more direct.

\subsection{The region algorithm}

Originally, the region algorithm was designed for online scheduling with and without commitment on a single machine. We extend it to unrelated machines by never migrating jobs between machines and per machine using the same design principles that guide the admission decisions of the region algorithm, as developed in~\cite{ChenEMSS2020}. Since we do not consider commitment in this section, we can significantly simplify the exposition of the region algorithm when compared to~\cite{ChenEMSS2020}. 

As in the previous section, a job is only processed by the machine it initially was started on. We say the job has been \emph{admitted} to this machine. 
Moreover, a running job can only be preempted by significantly smaller-size jobs, i.e., smaller by a factor of at least~$\frac{\eps}{4}$ with respect to the processing time, and a job~$j$ cannot start for the first time on machine~$i$ when its remaining slack is too small, i.e., less than~$\frac{\eps}{2}p_{ij}$.

Formally, at any time~$\tau$, the region algorithm maintains two sets of jobs: \textit{admitted jobs}, which have been started before or at time $\tau$, and \textit{available jobs}.  A job $j$ is available for machine~$i$ if it is released before or at time~$\tau$, is not yet admitted, and~$\tau$ is not too close to its deadline, i.e., 
$r_j \leq \tau$ and  $d_j - \tau \geq \big(1 + \frac\eps2\big) p_{ij}$. The intelligence of the \region lies in how it admits jobs. The actual scheduling decision then is simple
and independent of the admission of jobs: at any point in time and on each machine, schedule the shortest job that has been admitted to this machine and has not completed its processing time. In other words, we schedule admitted jobs on each machine in \textsc{Shortest Processing Time} (SPT) order. The \region never explicitly considers deadlines except when deciding whether to admit jobs. In particular, jobs can even be processed after their deadline. 

At any time~$\tau$, when there is a job~$j$ available for an \emph{idle} machine~$i$, i.e.,~$i$ is not processing any previously admitted job~$j'$, the shortest available job~$\newjob$ is immediately admitted to machine~$i$ at time $a_\newjob := \tau$. There are two events that trigger a decision of the \region: the release of a job and the completion of a job. If one of these events occurs at time~$\tau$, the \region invokes the {preemption} subroutine. This routine iterates over all machines and compares the processing time of the smallest job~$\newjob$ \textit{available} for machine~$i$ with the processing time of job~$j_i$ that is currently scheduled on machine~$i$. If~$p_{i\newjob} < \frac\eps4 p_{ij_i}$, job~$\newjob$ is admitted to machine~$i$ at time~$a_\newjob := \tau$ and, by the above scheduling routine, immediately starts processing. We summarize the region algorithm in Figure~\ref{alg:regionalg}.

\begin{figure}
\caption{Region algorithm}
\label{alg:regionalg}
\begin{algorithmic}
\STATE{\textbf{Scheduling Routine:} At any time $\tau$ and on any machine $i$, run the job with shortest processing time that has been admitted to $i$ and has not yet completed. \smallskip}
\STATE{\textbf{Event:} Release of a new job at time $\tau$}
	\STATE{\hspace{1em}Call Threshold Preemption Routine.\smallskip} 	
\STATE{\textbf{Event:} Completion of a job at time $\tau$}
	\STATE{\hspace{1em}Call threshold preemption routine.\smallskip} 
\STATE{\textbf{Threshold Preemption Routine:}}
\STATE{$i \assign 1$}
\STATE{$\newjob \assign$ a shortest job available for machine $i$ at $\tau$, i.e., $\newjob \in \arg\min\{p_{ij} \,|\, j \in \jobs, r_j \leq \tau \text{ and } d_j - \tau \geq (1 + \frac{\eps}{2}) p_{ij}\}$}
\WHILE{$i \leq m$}
	\STATE{$j \assign$ job processed on machine $i$ at time $\tau$}
	\IF{$j = \emptyset$}
		\STATE{admit job $\newjob$ to machine $i$}
		\STATE{call Threshold Preemption Routine}
	\ELSIF{$p_{i\newjob} < \frac\eps4 p_{ij}$}
		\STATE{admit job $\newjob$ to machine $i$}
		\STATE{call Threshold Preemption Routine}
	\ELSE
		\STATE{$i \assign i+1$}
		\STATE{$\newjob \assign$ a shortest job available for machine $i$ at $\tau$, i.e., $\newjob \in \arg\min\{p_{ij} \,|\, j \in \jobs, r_j \leq \tau \text{ and } d_j - \tau \geq (1 + \frac{\eps}{2}) p_{ij}\}$}
	\ENDIF
\ENDWHILE
\end{algorithmic}
\end{figure}

The proof of the analysis splits again naturally into two parts: The first part is to show that the region algorithm completes at least half of all admitted jobs, and the second is to use \cref{theo:com:AdmitEnoughJobs} to compare the number of admitted jobs to the throughput of an optimal non-migratory algorithm.

\subsection{Completing sufficiently many admitted jobs}

The main result of this section is the following theorem. 

\begin{theorem}\label{theo:otm:CompleteHalf}
	Let $0< \eps \leq 1$. Then the \region completes at least half of all admitted jobs before their deadline.
\end{theorem}

The proof of \cref{theo:otm:CompleteHalf} relies on two technical results that enable us to restrict to instances with one machine and further only consider jobs that are admitted by the \region in this instance. Then, we can use the analysis of the region algorithm in~\cite{ChenEMSS2020} to complete the proof. 

We start with the following observation. Let~$\I$ be an instance of online throughput maximization with the job set~$\jobs$ and let~$J\subseteq \jobs$ be the set of jobs admitted by the \region at some point. It is easy to see that a job~$j \notin J$ does not influence the scheduling or admission decisions of the \region. The next lemma formalizes this statement and follows immediately from the just made observations. 

\begin{lemma}\label{lem:otm:RestrictJobs}
	For any instance~$\I$ for which the \region admits the job set~$J\subseteq \jobs$, the reduced instance~$\I'$ containing only the jobs~$J$ forces the \region with consistent tie breaking to admit all jobs in~$J$ and to create the same schedule as produced for the instance~$\I$. 
\end{lemma}

The proof of the main result compares the number of jobs finished on time,~$F \subseteq J$, to the number of jobs unfinished by their respective deadlines,~$U = J \setminus F$. To further simplify the instance, we use that the \region is non-migratory and restrict to single-machine instances. To this end, let~$F^{(i)}$ and~$U^{(i)}$ denote the finished and unfinished, respectively, jobs on machine~$i$.

\begin{lemma}\label{lem:otm:ReduceToOneMachine}
	Let~$i \in \{1,\ldots,m\}$. There is an instance~$\I'$ on one machine with job set~$\jobs' = F^{(i)} \cup U^{(i)}$. Moreover, the schedule of the \region for instance~$\I'$ with consistent tie breaking is identical to the schedule of the jobs~$\jobs'$ on machine~$i$. In particular,~$ F'= F^{(i)}$ and~$U'= U^{(i)}$. 
\end{lemma}

\begin{proof}
	By \Cref{lem:otm:RestrictJobs}, we can restrict to the jobs admitted by the \region. Hence, let~$\I$ be such an instance with~$F^{(i)}\cup U^{(i)}$ being admitted to machine~$i$. 
	As the \region is non-migratory, the sets of jobs scheduled on two different machines are disjoint. Let~$\I'$ consist of the jobs in~$\jobs':= F^{(i)} \cup U^{(i)}$ and one machine. We set~$p_j'= p_{ij}$ for~$j \in \jobs'$. The \region on instance~$\I$ admits all jobs in~$\jobs$. In particular, it admits all jobs in~$\jobs'$ to machine~$i$. 
	
	We inductively show that the schedule for the instance~$\I'$ is identical to the schedule on machine~$i$ in instance~$\I$. To this end, we index the jobs in~$\jobs'$ in increasing admission time points in instance~$\I$. 
	
	It is obvious that job~$1 \in \jobs'$ is admitted to the single machine at its release date~$r_1$ as happens in instance~$\I$ since the \region uses consistent tie breaking. Suppose that the schedule is identical until the admission of job~$\newjob$ at time~$a_\newjob=\tau$. If~$\newjob$ does not interrupt the processing of another job, then~$\newjob$ will be admitted at time~$\tau$ in~$\I'$ as well. Otherwise, let~$j \in \jobs'$ be the job that the \region planned to process at time~$\tau$ \emph{before} the admission of job~$\newjob$. Since~$\newjob$ is admitted at time~$\tau$ in~$\I$,~$\newjob$ is available at time~$\tau$, satisfies~$p_{\newjob}' = p_{i\newjob} < \frac\eps4 p_{ij}  = \frac\eps4 p_j'$, and did not satisfy both conditions at some earlier time~$\tau'$ with some earlier admitted job~$j'$. Since the job set in~$\I'$ is a subset of the jobs in~$\I$ and we use consistent tie breaking, no other job~$j^* \in \jobs'$ that satisfies both conditions is favored by the \region over~$\newjob$. Therefore, job~$\newjob$ is also admitted at time~$\tau$ by the \region in instance~$\I'$. Thus, the schedule created by the \region for~$\jobs'$ is identical to the schedule of~$\jobs$ on machine~$i$ in the original instance. 
\end{proof}

For proving \cref{theo:otm:CompleteHalf}, we consider a worst-case instance for the \region where ``worst'' is with respect to the ratio between admitted and successfully completed jobs. 
Since the \region is non-migratory, there exists at least one machine in such a worst-case instance that ``achieves'' the same ratio as the whole instance. By the just proven lemma, we can find a worst-case instance on a single machine. However, on a single machine, the \region algorithm in this paper is identical to the algorithm designed in~\cite{ChenEMSS2020}. Therefore, we simply follow the line of proof developed in~\cite{ChenEMSS2020} to show \cref{theo:otm:CompleteHalf}. 

More precisely, in~\cite{ChenEMSS2020} we show that the existence of a late job~$j$ implies that the the set of jobs admitted by~$j$ or by one of its children contains more finished than unfinished jobs. Let~$F_j$ denote the set of jobs admitted by~$j$ or by one of its children that {\em finish on time}. Similarly, we denote the set of such jobs that complete after their deadlines, i.e., that are {\em unfinished at their deadline}, by~$U_j$.
We restate the following lemma, which was originally shown in a single-machine environment but clearly also holds for unrelated machines.

\begin{lemma}[Lemma 3 in \cite{ChenEMSS2020}]\label{lem:otm:F>U}
	Consider some job $j$ admitted to some machine $i\in\{1,\dots,m\}$. If~$C_j - a_j \geq (\ell + 1) p_{ij}$ for~$\ell > 0$, then~$|F_j| - |U_j| \geq \lfloor \frac{4\ell}{\eps}\rfloor$. 
\end{lemma}

\begin{proof}[Proof of \cref{theo:otm:CompleteHalf}]
	Let $U$ be the set of jobs that are unfinished by their deadline but whose ancestors have all completed on time. Every job $j\in U$ was admitted by the algorithm at some time $a_j$
	with $d_j - a_j \geq \big(1+\frac\eps2\big) p_{ij}$. Since~$j$ is unfinished, we have~$C_j - a_j > d_j - a_j \geq \big(1+\frac{\eps}{2}\big) p_{ij}$. By \Cref{lem:otm:F>U}, $|F_j| - |U_j| \geq \big\lfloor \frac{4\cdot \eps/2}{\eps} \big\rfloor = 2 $. Thus, 
	\[
	|F_j| + |U_j| \leq 2|F_j| - 2 < 2|F_j|.
	\]
	Since every ancestor of such a job~$j$ finishes on time, this completes the proof. 
\end{proof}

\subsection{The region algorithm admits sufficiently many jobs}

In this section, we show the following theorem and give the proof of \cref{theo:otm:UB-no}.  

\begin{theorem}\label{theo:otm:AdmitEnoughJobs}
	An optimal non-migratory (offline) algorithm completes at most a factor~$\big(\frac{8}{\eps} + 4\big)$ more jobs on time than admitted by the \region. 
\end{theorem}

\begin{proof}
	As in the previous section, fix an instance and an optimal solution \opt. Let~$X$ be the set of jobs in \opt that the \region did not admit. We assume without loss of generality that all jobs in \opt finish on time. Further, let~$J$ denote the set of jobs that the \region admitted. Then, $X\cup J$ is a superset of the jobs in \opt. Thus,~$|X|\leq \big(\frac{8}{\eps} + 3 \big) |J|$ implies \Cref{theo:otm:AdmitEnoughJobs}. 
	
	Consider an arbitrary but fixed machine~$i$. We compare again the throughput of the optimal schedule on machine~$i$ to the throughput of the \region on machine~$i$. 
	Let~$\barx \subseteq X$ denote the jobs in \opt scheduled on machine~$i$ and let~$\barj$ denote the jobs scheduled by the \region on machine~$i$. Then, showing~$|\barx| \leq \big(\frac{8}{\eps} + 3\big) |\barj|$ suffices to prove the main result of this section. Given that the \region satisfies Properties~\ref{enum:alg:available} and~\ref{enum:alg:threshold}, \cref{theo:otm:Charging} already provides a bound on the cardinality of~$\barx$ in terms of the \emph{intervals} corresponding to the schedule on amchine~$i$. Thus, it remains to show that the \region indeed qualifies for applying \cref{theo:otm:Charging} and that the bound developed therein can be translated to a bound in terms of~$|\barj|$. 
	
	We start by showing that the \region satisfies the assumptions necessary for applying \cref{theo:otm:Charging}. Clearly, as the \region only admits a job~$j$ at time~$\tau$ if~$d_j - \tau \geq \big(1+\frac\eps2\big) p_{ij}$, setting~$\delta = \frac\eps2$ proves that the \region satisfies~\ref{enum:alg:available}. For~\ref{enum:alg:threshold}, we retrospectively analyze the schedule generated by the \region. For a time~$\tau$, let~$j_i$ denote the job scheduled on machine~$i$. Then, setting~$u_{i,\tau} := \frac\eps4 p_{ij_i}$ or~$u_{i,\tau} = \infty$ if no such job~$j_i$ exists, indeed provides us with the machine-dependent threshold necessary for~\ref{enum:alg:threshold}. This discussion also implies that~$u^{(i)}$ has only countably many points of discontinuity as there are only finitely many jobs in the instance, and that~$u^{(i)}$ is right-continuous. 
	
	Hence, let~$\bari$ denote the set of maximal intervals~$I_t = [\tau_t, \tau_{t+1})$ for~$t \in [T]$ of constant threshold~$u_{i\tau}$. Thus, by \cref{theo:otm:Charging},
	\begin{equation}\label{eq:otm:Charging}
	|\barx| \leq \sum_{t=1}^{T} \epsdeltafrac \frac{\tau_{t+1} - \tau_t}{\thresh}  + T.
	\end{equation}
	
	As the threshold~$u_{i,\tau}$ is proportional to the processing time of the job currently scheduled on machine~$i$, the interval~$I_t$ either represents an idle interval of machine~$i$ (with~$u_{i\tau} = \infty$) or corresponds to the uninterrupted processing of some job~$j$ on machine~$i$. We denote this job by~$j_t$ if it exists. We consider now the set~$\bari_j \subseteq \bari$ of intervals with~$j_t = j$ for some particular job~$j \in \barj$. As observed, these intervals correspond to job~$j$ being processed which happens for a total of~$p_{ij}$ units of time. Combining with~$\thresh = \frac\eps4 p_{ij}$ for~$I_t \in \bari_j$, we get
	\[
	\sum_{t: I_t \in \bari_j} \frac{\tau_{t+1} - \tau_t}{\thresh} = \frac{p_{ij}}{\frac\eps4 p_{ij}} = \frac4\eps.
	\]
	As~$\delta = \frac{\eps}{2}$, we additionally have that~$\epsdeltafrac = 2$. Hence, we rewrite Equation~\eqref{eq:otm:Charging} by 
	\[
	|\barx| \leq \frac8\eps |\barj| + T.
	\]
	
	It remains to bound~$T$ in terms of~$|\barj|$ to conclude the proof. To this end, we recall that the admission of a job~$j$ to a machine interrupts the processing of at most one previously admitted job. Hence, the admission of~$|\barj|$ jobs to machine~$1$ creates at most~$2|\barj|+1$ intervals. 
	
	If the \region does not admit any job to machine~$i$, i.e.,~$|\barj| = 0$, then~$u_{i\tau} = \infty$ for each time point~$\tau$. Hence, there exists no job scheduled on machine~$i$ by \opt that the \region did not admit. In other words, $X_i = \emptyset$ and~$|X_i| = 0 = |J_i|$. Otherwise,~$ 2|\barj|+1 \leq 3 |\barj|$. 
	Therefore, 
	\[
	|\barx| \leq \left(\frac8\eps + 3\right) |\barj|.
	\]
	Combining with the observation about~$\barx$ and~$\barj$ previously discussed, we obtain
	\[
	|\opt| \leq |X \cup J| = \sum_{i=1}^m |X_i| + |J| \leq \left(\frac8\eps + 3\right) \sum_{i=1}^m |\barj| + |J| = \left(\frac8\eps + 4\right) |J|,
	\]
	which concludes the proof. 
\end{proof}

\subsection{Finalizing the proof of \Cref{theo:otm:UB-no}}

\begin{proof}[Proof of Theorem~\ref{theo:otm:UB-no}]
	In \Cref{theo:otm:CompleteHalf} we show that the \region completes at least half of all admitted jobs~$J$ on time. In \Cref{theo:otm:Charging}, we bound the throughput~$|\opt|$ of an optimal non-migratory solution by $\big(\frac{8}{\eps}+4\big)|J|$. Combining these theorems shows that the \region achieves a competitive ratio of~\(	c = 2 \cdot \big( \frac8\eps  + 4 \big) = \frac{16}{\eps} + 8. \)
\end{proof}

\section{Conclusion} 

In this paper, we close the problem of online single-machine throughput maximization with and without commitment requirements. For both commitment settings, we give an optimal online algorithm. Further, our algorithms run in a multiple-machine environment, even on heterogenous machines. Our algorithms compute non-migratory schedules on unrelated machines with the same competitive ratio $\OO\big(\frac1\eps\big)$ as for a single machine and improve substantially upon the state of the art.

It remains open whether the problem with a large number of machines admits an online algorithm with a better competitive ratio. Indeed, very recently, Moseley et al.~\cite{MoseleyPSZ21} present an $\OO(1)$-competitive algorithm for $m\geq 2$ parallel identical machines without commitment. It remains, however, unclear if this result can be lifted to other machine environments or any of the commitment models.

There are other examples in the literature in which the worst-case ratio for a scheduling problem improves with an increasing number of machines. Consider, e.g., the non-preemptive offline variant of our throughput maximization problem on identical machines. There is an algorithm with approximation ratio of $1.55$ for any $m$ which is improving with increasing number of machines, converging to $1$ as  $m$ tends to infinity \cite{ImLM17}. The second part of the result also holds for the weighted problem.

Another interesting question asks whether randomization allows for improved results. Recall that there is an $\OO(1)$-competitive randomized algorithm for scheduling on a single machine without commitment and without slack assumption~\cite{KalyanasundaramP03}.
Therefore is seems plausible that randomization also helps designing algorithms with improved competitive ratios for the different commitment models, for which only weak lower bounds are known~\cite{ChenEMSS2020}.

Further, we leave migratory scheduling on unrelated machines as an open problem. Allowing migration in this setting means that, on each machine $i$, a certain fraction of the processing time $p_{ij}$ is executed, and these fractions must sum to one. Generalizing the result we leverage for identical machines~\cite{KalyanasundaramP01}, it is conceivable that any migratory schedule can be turned into a valid non-migratory schedule of the same jobs by adding a constant number of machines \emph{of each type}. Such a result would immediately allow to transfer our competitive ratios to the migratory setting (up to constant factors). Devanur and Kulkarni~\cite{DevanurK18} show a weaker result that utilizes speed rather than additional machines. Note that the strong impossibility result of Im and Moseley~\cite{ImM16} does not rule out the desired strengthening because we make the $\varepsilon$-slack assumption for every job and machine eligible for it. Further, we -- as well as Devanur and Kulkarni~\cite{DevanurK18} -- assume that the processing time of each job $j$ satisfies $p_{ij}\leq d_j-r_j$ on any eligible machine $i$, whereas the lower bound in \cite{ImM16} requires jobs that violate this reasonable assumption.

Further research directions include generalizations such as weighted throughput maximization. While strong lower bounds exist for handling weighted throughput with commitment~\cite{ChenEMSS2020}, there remains a gap for the problem without. The known lower bound of~$\Omega\big(\frac1\eps\big)$ already holds for unit weights~\cite{ChenEMSS2020}. A natural extension of the region algorithm bases its admission decisions on the density, i.e., the ratio of the weight of a job to its processing time. The result is an algorithm similar to the~$\OO\big(\frac1{\eps^2} \big)$-competitive algorithm by Lucier et al.~\cite{LucierMNY13}. Both algorithms only admit available jobs and interrupt currently running jobs if the new job is denser by a certain factor. However, we can show that there is a lower bound of~$\Omega\big(\frac1{\eps^2} \big)$ on the competitive ratio of such algorithms. Hence, in order to improve the upper bound for online weighted throughput maximization, one needs to develop a new type of algorithm. 

\newpage
\bibliographystyle{abbrv} 
\bibliography{throughput}
\end{document}